\pdfoutput=1
\RequirePackage{ifpdf}
\ifpdf 
\documentclass[pdftex]{sigma}
\else
\documentclass{sigma}
\fi

\numberwithin{equation}{section}

\newtheorem{Theorem}{Theorem}[section]
\newtheorem*{Theorem*}{Theorem}
\newtheorem{Corollary}[Theorem]{Corollary}
\newtheorem{Lemma}[Theorem]{Lemma}

{ \theoremstyle{definition}
	\newtheorem{Definition}[Theorem]{Definition}
	
	\newtheorem{Example}[Theorem]{Example}
	 }

\usepackage{tikz}
\usetikzlibrary{matrix,arrows,patterns,cd}
\usepackage{leftidx}

\usepackage[mathscr]{eucal}

\newcommand{\II}{{\boldsymbol{1}}}

\newcommand{\CC}{{\mathbb C}}

\newcommand{\RR}{{\mathbb R}}

\newcommand{\NN}{{\mathbb N}}

\newcommand{\CoinX}[1]{C_0^\infty\big({#1}\big)}
\newcommand{\GoinX}[1]{\Gamma_0^\infty({#1})}

\newcommand{\fc}{\mathrm{fc}}
\newcommand{\pc}{\mathrm{pc}}
\newcommand{\sfc}{\mathrm{sfc}}
\newcommand{\spc}{\mathrm{spc}}

\newcommand{\DD}{{\mathscr D}}
\newcommand{\EE}{{\mathscr E}}

\newcommand{\LL}{{\mathcal L}}

\newcommand{\Ran}{\operatorname{Ran}}
\DeclareMathOperator{\carr}{carr}
\DeclareMathOperator{\supp}{supp}

\renewcommand{\Re}{\operatorname{Re}}

\newcommand{\ip}[2]{{\big\langle #1 \, \big|\, #2\big\rangle}}
\newcommand{\ket}[1]{{\vert #1\rangle}}
\newcommand{\bra}[1]{{\langle #1 \vert}}

\newcommand{\Sf}{{\mathscr S}}

\newcommand{\Uf}{{\mathscr U}}

\DeclareMathOperator{\Sol}{Sol}

\newcommand{\id}{{\rm id}}

\DeclareMathOperator{\inte}{int}

\newcommand{\loc}{{\rm loc}}

\makeatletter
\newcommand{\raisemath}[1]{\mathpalette{\raisem@th{#1}}}
\newcommand{\raisem@th}[3]{\raisebox{#1}{$#2#3$}}
\makeatother

\newcommand{\adj}[1]{\leftidx{^{t}}{#1}{}}

\renewcommand{\theenumi}{(\alph{enumi})}

\newcommand{\dd}{\mathrm{d}}
\DeclareMathOperator{\Ind}{Index}

\newcommand{\Ps}{\mathcal{P}}

\begin{document}
\allowdisplaybreaks

\renewcommand{\thefootnote}{}

\newcommand{\arXivNumber}{2303.02993}

\renewcommand{\PaperNumber}{057}

\FirstPageHeading
	
\ShortArticleName{Modified Green-Hyperbolic Operators}
	
\ArticleName{Modified Green-Hyperbolic Operators\footnote{This paper is a~contribution to the Special Issue on Global Analysis on Manifolds in honor of Christian B\"ar for his 60th birthday. The~full collection is available at \href{https://www.emis.de/journals/SIGMA/Baer.html}{https://www.emis.de/journals/SIGMA/Baer.html}}}
	
\Author{Christopher J.~FEWSTER}
	
\AuthorNameForHeading{C.J.~Fewster}
	
\Address{Department of Mathematics,
		University of York, Heslington, York YO10 5DD, UK}
	\Email{\href{mailto:chris.fewster@york.ac.uk}{chris.fewster@york.ac.uk}}
	\URLaddress{\url{https://www.york.ac.uk/maths/people/chris-fewster/}}
	
\ArticleDates{Received April 20, 2023, in final form July 29, 2023; Published online August 08, 2023}
		
\Abstract{Green-hyperbolic operators -- partial differential operators on globally hyperbolic spacetimes that (together with their formal duals) possess advanced and retarded Green operators -- play an important role in many areas of mathematical physics. Here, we study modifications of Green-hyperbolic operators by the addition of a possibly nonlocal operator acting within a compact subset $K$ of spacetime, and seek corresponding `$K$-nonlocal' generalised Green operators. Assuming the modification depends holomorphically on a parameter, conditions are given under which $K$-nonlocal Green operators exist for all parameter values, with the possible exception of a discrete set. The exceptional points occur precisely where the modified operator admits nontrivial smooth homogeneous solutions that have past- or future-compact support. Fredholm theory is used to relate the dimensions of these spaces to those corresponding to the formal dual operator, switching the roles of future and past. The $K$-nonlocal Green operators are shown to depend holomorphically on the parameter in the topology of bounded convergence on maps between suitable Sobolev spaces, or between suitable spaces of smooth functions. An application to the LU factorisation of systems of equations is described.}
		
		
\Classification{35R01; 35R09; 47B93}

\begin{flushright}
\begin{minipage}{65mm}
\em Dedicated to Christian B\"ar\\ on the occasion of his sixtieth birthday
\end{minipage}
\end{flushright}

\renewcommand{\thefootnote}{\arabic{footnote}}
\setcounter{footnote}{0}

 \section{Introduction}

 Linear partial differential operators are the workhorses of mathematical physics, providing the
 simplest models of classical and quantum field theories from which more complicated interacting models may be built. In general relativity or other nonlinear theories, linear operators appear whenever the theory is linearised, for example, to study the stability of solutions, or the propagation of gravitational waves.

 A particularly useful general class of operators acting between spaces of smooth sections of vector bundles over globally hyperbolic Lorentzian manifolds has been introduced by B\"ar~\cite{Baer:2015} under the name `Green-hyperbolic operators'. Green-hyperbolicity is a generalisation, rather than a specialisation, of hyperbolicity: a Green-hyperbolic operator need not be hyperbolic, and there are examples that are elliptic, or of indefinite type. The defining property of a Green-hyperbolic operator is that it should possess advanced and retarded Green operators, along with its formal dual, and from this simple algebraic requirement many other properties flow, as described elegantly by B\"ar. In particular, the Green operators are unique, continuous, and have extensions that are continuous inverses to
 (extensions of) the Green-hyperbolic operator on various spaces of smooth or distributional sections.

 Examples of Green-hyperbolic operators include (a) normally hyperbolic operators, such as the
 wave operator $\Box$ and its variants allowing for the inclusion of potentials and external vector potentials, (b) first order symmetric hyperbolic systems on
 globally hyperbolic spacetimes~\cite{Baer:2015}; (c)~any operator whose square is Green-hyperbolic, thus incorporating Dirac type operators; (d)~operators whose solution theory may be related to one of the above, such as the (non-hyperbolic) Proca operator $-\delta\dd+m^2$, whose Green operators are obtained from those of ${-(\delta\dd+\dd\delta)+m^2}$ on smooth $1$-form fields. See~\cite{Baer:2015} for these and other examples.

 The purpose of this paper is to study modifications of Green-hyperbolic operators, that can lead outside the class of partial differential operators. For simplicity, we mainly study the case of scalar operators but there is no hindrance (beyond those of notation!) to extending our results to the general bundle case. The operators we consider are of the form $P+A$, where $P$ is Green-hyperbolic and $A$ is a continuous linear self-map of $C^\infty(M)$ (not necessarily a differential operator) whose range is contained in $C^\infty_K(M)$, the smooth functions supported in a compact subset $K\subset M$. Without loss of generality, we may always assume that $K$ is \emph{topologically regular}, that is, equal to the closure of its interior.
 An operator of this type is potentially nonlocal, though the nonlocality is, as it were, localised within $K$. By the kernel theorem, any such operator can be represented as
 \begin{equation*}
 	A\phi = \int_M T(x,y)\phi(y)\mu_y,
 \end{equation*}
 where $\mu$ is a smooth density, and $T\in \DD'(M\times M)$ has support in $K\times M$ and is semi-regular in its first slot
 \big(i.e., $T$ belongs to the nuclear tensor product $C^\infty(M)\widehat{\otimes} \DD'(M)$\big). Operators of this type include, but go beyond, differential and pseudodifferential operators supported in $K$ (for which the singular support of $T$ is confined
 to the diagonal).

 There are several applications for operators of this type. In perturbative algebraic quantum field theory (pAQFT) they arise from the class of regular interactions, and the results proved here establish the existence of suitable Green operators needed in~\cite{HawRej:2020}, for example. Another application is to noncommutative potential scattering~\cite{LechnerVerch:2015}, where equations such as
 \begin{equation*}
 	P \phi + w\star \phi = 0
 \end{equation*}
 appear as toy models for the dynamics of classical and quantum fields on a noncommutative spacetime. Here $P$ is a Green-hyperbolic operator, $w$ is a fixed smooth function, and $\star$ is a noncommutative deformation
 of multiplication, differing only from pointwise multiplication inside a compact set $K$~\cite{LechnerWaldmann:2016}. The application of the results obtained here to such models will be discussed elsewhere~\cite{FewVer:20xxb}. See also~\cite{DappiaggiFinster:2020} for a slightly different context in which
 nonlocal equations have appeared recently.

 One may in fact develop an entire theory of nonlocal Green-hyperbolic operators and this is done in the companion paper~\cite{FewVer:20xx}. The purpose of this paper is to investigate the technical issue of the existence and properties of (a suitably generalised concept of) Green operators $E^\pm_{P+A(\lambda)}$ for $P+ A(\lambda)$ under suitable conditions on $P$ and $A(\lambda)$ for $\lambda\in\CC$. A particular focus will be on the analytic dependence of the resulting Green operators on $\lambda$ within suitable locally convex spaces and on a suitable domain in $\CC$. These results are useful even in situations where the operators $P+A(\lambda)$ are local Green-hyperbolic operators. For example, they have been applied in~\cite{FewsterJubbRuep:2022} in the context of measurement schemes for observables in QFT. The Green operators we study have properties similar to those of Green-hyperbolic operators, with the following generalised support property:
 \begin{equation*} \supp E_{P+A(\lambda)}^{\pm} f \subset
 	\begin{cases} J^{\pm}(\supp f), & J^\pm(\supp f)\cap K=\varnothing,\\
 		J^{\pm}(\supp f\cup K), & \text{otherwise},
 	\end{cases}
 \end{equation*}
 for compactly supported $f$, where $J^{+/-}(S)$ are the causal future/past of a set $S$. This support property characterises
 what we call \emph{$K$-nonlocal Green operators}, set out precisely in Definition~\ref{def:KnonlocalGops} below.

 The main result of this paper is Theorem~\ref{thm:GHpert}, which sets out conditions
 on $P$ and $A(\lambda)$ under which suitable Green operators for $P+A(\lambda)$ exist.
 The principal hypothesis on $P$ is that it is a Green-hyperbolic operator whose Green operators
 have extensions to continuous maps between the Sobolev spaces $H_0^s(M)$ and $H^{s+\beta}_\loc(M)$ for all sufficiently large $s$ and some fixed $\beta$. This hypothesis is valid for second order normally hyperbolic operators with $\beta=1$ (see~\cite[Theorem 6.5.3]{DuiHoer_FIOii:1972}). The main hypothesis on $A(\lambda)$ is that each $A(\lambda)$ is a continuous linear self-map of $C^\infty(M)$ with continuous extensions mapping between Sobolev spaces $H^{s}_\loc(M)$ to~$H^{s+\gamma}_K(M)$ for all sufficiently large $s$ and some fixed $\gamma>-\beta$. This implies that the compositions $A(\lambda)E^\pm$
 are compact maps between $H^s_0(M)$ and $H^s_K(M)$ and it is required that they depend holomorphically on $\lambda$ in the topology of bounded convergence on linear maps between these spaces. A brief summary of the various topological spaces and topologies used in this work is provided in Appendix~\ref{appx:topspaces}.

 Given the above assumptions, the analytic Fredholm theorem can be used to
 find inverses $\big(I+A(\lambda)E^\pm\big)^{-1}$ on $H^s_K(M)$ for all sufficiently large $s\in\RR$ and all complex $\lambda$ in an open neighbourhood of zero,
 whose (possibly empty) complement is a discrete subset $S$ of $\CC$. Exceptional values $\lambda\in S$ occur precisely when there exist nontrivial \emph{smooth} solutions to
 $(P+A(\lambda))\phi = 0$,
 whose support is either past- or future-compact, i.e., nontrivial solutions that vanish identically
 at early or late times, representing spontaneously appearing or disappearing disturbances; they are clearly excluded if any sort of energy estimate is available. For $\lambda\in\CC\setminus S$, one may use the inverses to construct $K$-nonlocal Green operators for $P+A(\lambda)$. It is also proved that the resulting $K$-nonlocal Green operators are holomorphic on $\CC\setminus S$, with respect to the topology of bounded convergence on linear maps between $C_0^\infty(M)$ and $C^\infty(M)$.
 The power series expansion of the Green operators about $\lambda=0$ corresponds to a Born expansion of the Green operators. We also show how classical M{\o}ller operators~\cite{DuetschFredenhagen:2003,HawRej:2020} can be constructed to relate the free and interacting dynamics.

 In the situation where our main hypotheses are also satisfied for the formal duals of $P$ and~$A(\lambda)$, we apply Fredholm index theory to show that the dimension of the space of spontaneously appearing (resp., disappearing) solutions for $P+A(\lambda)$ is equal to the dimension of the space of spontaneously disappearing (resp., appearing) solutions for its formal dual,
 \begin{equation}\label{eq:Fred}
 	\dim\ker (P+A(\lambda))|_{C^\infty_{\pc/\fc}} = \dim\ker \big(\adj{P}+\adj{A}(\lambda)\big)|_{C^\infty_{\fc/\pc}},
 \end{equation}
 for all $\lambda\in \CC$. Here, the subscripts pc and fc denote past-compact and future-compact support, respectively. Consequently, the spaces of appearing and disappearing solutions of a formally self-dual operator have equal dimension.

 We mention that Dappiaggi and Finster~\cite{DappiaggiFinster:2020} have recently studied
 nonlocal equations arising from `causal variational principles'. These differ from the equations we consider because the nonlocality need not be confined to a compact spacetime region. Their existence proofs are also different: they assert conditions under which a variant of the classical energy estimates can be formulated and then follow essentially the classical arguments used to establish existence for Cauchy problems of normally hyperbolic partial differential equations.

 The present paper is structured as follows: in Section~\ref{sec:GHOs}, we recall the definition and main properties of Green-hyperbolic operators and develop the appropriate notion of $K$-nonlocal Green operators. Section~\ref{sec:main} contains the statement of our main result, Theorem~\ref{thm:GHpert} and illustrates it with examples of how it may be used and of the necessity of some of its hypotheses. Section~\ref{sec:LU} provides an application of our result to the LU factorisation and solution of certain systems of nonlocal equations. The main result is proved in Section~\ref{sec:proof}, by a sequence of results initially in Sobolev spaces and then for smooth functions, while~\eqref{eq:Fred} is proved in Section~\ref{sec:Fredholm} using Fredholm theory and the results of Section~\ref{sec:proof}. Appendix~\ref{appx:topspaces} collects some necessary background on the topological spaces and topologies appearing in the text.

 \section{Green-hyperbolic operators}\label{sec:GHOs}

{\bf Preliminaries.}
 We begin by recalling the general setting of Green-hyperbolic operators~\cite{Baer:2015}. Let~$M$ be a smooth finite-dimensional manifold, allowing the possibility that $M$ has finitely many connected components with possibly different dimensions, and let $g$ be a smooth Lorentzian metric on $M$ of signature $+--\cdots$. With these structures, $M$ is automatically Hausdorff and paracompact~\cite{Geroch_spinorI:1968}.
 We assume that $(M,g)$ is time-orientable and that a time-orientation has been chosen. To minimise notation, we denote the Lorentzian spacetime formed by the manifold, metric and time orientation with the single symbol $M$. The volume measure induced by the metric will be denoted $\mu$. On other points of notation,
 the symbol $\subset$ will always allow for the possibility of equality, while $\NN_0$ and $\NN$ denote the natural numbers with or without zero, respectively.

 As usual, the causal future/past of a point $x\in M$ is denoted $J^\pm(x)$ and comprises all points (including $x$) that may be reached from $x$ along smooth future/past-directed curves. (Throughout this paper, we tacitly order alternatives labelled by $\pm$ or $\mp$ so that the alternative labelled by the upper symbol comes first.)
 If $S\subset M$ then one writes $J^\pm(S)=\cup_{x\in S} J^\pm(x)$, and $J(S)=J^+(S)\cup J^-(S)$. The spacetime is globally hyperbolic if it contains no causal curves and $J^+(K)\cap J^-(K)$ is compact for all compact sets $K\subset M$~\cite{Bernal:2006xf}. Globally hyperbolic spacetimes can be foliated into smooth spacelike Cauchy surfaces~\cite{Bernal:2004gm}; it is also the case that $J^\pm(K)$ are closed whenever $K$ is compact. We adopt the following terminology: $S\subset M$ is
 \emph{spacelike compact} if $S$ is closed and $S\subset J(K)$ for some compact $K$; $S$ is \emph{future/past-compact} if $S\cap J^{\pm}(x)$ is compact for all $x\in M$; $S$ is \emph{strictly future/past-compact} if
 $S\subset J^{\mp}(K)$ for some compact $K$.

 If $B$ is a vector bundle with finite-dimensional fibres, then $\Gamma^\infty(B)$ will denote the corresponding space of smooth sections and $\Gamma_{0/\pc/\fc/\spc/\sfc/\mathrm{sc}}^\infty(B)$ will be the spaces of smooth sections with com\-pact/past-compact/future-compact/strictly past-compact/strictly future-compact/space\-li\-ke-compact support. The space of smooth sections with support contained in some a closed subset $A\subset M$ is denoted $\Gamma^\infty_A(M)$. Further details on the topologies of these spaces are summarised in Appendix~\ref{appx:topspaces}. We will only consider bundles with finite-dimensional fibres that are (without any real loss) vector spaces over $\CC$. The bilinear (not sesquilinear!) pairing between sections of dual bundle $B^*$ and sections of $B$ is denoted with angle brackets $\langle \cdot,\cdot\rangle \colon \Gamma^\infty(B^*)\times \Gamma^\infty(B)\to C^\infty(M)$.

{\bf Green-hyperbolicity.} Suppose that $B_1$ and $B_2$ are bundles over $M$ and that $P\colon \Gamma^\infty(B_1)\to \Gamma^\infty(B_2)$ is a linear partial differential operator. Then there is a formal dual $\adj{P}\colon \Gamma^\infty(B_2^*)\to \Gamma^\infty(B_1^*)$ given by
 \begin{equation*}
 	\int_M \mu\, \big\langle \adj{P}f, \phi\big\rangle = \int_M\mu\, \langle f, P\phi\rangle
 \end{equation*}
 for all $\phi\in \Gamma^\infty(B_1)$ and $f\in \Gamma^\infty(B_2^*)$ whose supports intersect compactly; $\adj{P}$ is also a linear partial differential operator.
 If they exist, linear maps $E^\pm\colon \GoinX{B_2}\to \Gamma^\infty(B_1)$ obeying
 \begin{itemize}\itemsep=0pt
 	\item[(G1)] $E^\pm P f=f$ for all $f\in\GoinX{B_1}$,
 	\item[(G2)] $P E^\pm f = f$ for all $f\in\GoinX{B_2}$,
 	\item[(G3)] $\supp E^\pm f \subset J^\pm (\supp f)$ for all $f\in\GoinX{B_2}$
 \end{itemize}
 are called advanced ($-$) and retarded ($+$) Green operators for $P$.\footnote{B\"ar reverses the usage of `advanced' and `retarded'; we adopt the more standard convention.}
 \begin{Definition}
 	$P$ is said to be \emph{Green-hyperbolic} if both $P$ and $\adj{P}$ admit advanced and retarded Green operators.
 \end{Definition}

 It is a remarkable fact that this purely algebraic definition -- only requiring linearity of $E^\pm$ and with no presumption of uniqueness or continuity -- has the following consequence, which is a summary of results in~\cite[Sections 3 and 4]{Baer:2015} and deserves to be called the \emph{first main theorem of Green-hyperbolicity}.
 \begin{Theorem}\label{thm:GHbasics}
 	Let $P\colon\Gamma^\infty(B_1)\to \Gamma^\infty(B_2)$ be a Green-hyperbolic operator and consider:
 	\begin{enumerate}\renewcommand{\theenumi}{(\roman{enumi})}\itemsep=0pt
 		\item[$(i)$] the restriction $P\colon\Gamma^\infty_{\spc/\sfc}(B_1)\to\Gamma^\infty_{\spc/\sfc}(B_2)$,
 		\item[$(ii)$] the restriction $P\colon\Gamma^\infty_{\pc/\fc}(B_1)\to\Gamma^\infty_{\pc/\fc}(B_2)$, or
 		\item[$(iii)$] the extension $P\colon\DD'_{\pc/fc}(B_1)\to\DD'_{\pc/fc}(B_2)$ $($see below$)$.
 	\end{enumerate}
 	Then, in each case, $P$ has continuous inverses, denoted
 	$\tilde{E}^\pm$, $\overline{E}^\pm$ and $\widehat{E}^\pm$ in cases
 	$(i)$, $(ii)$, $(iii)$ respectively, and which are successive extensions of $E^\pm$. Each of these inverses has the support property $(G3)$, replacing
 	$\Gamma_0^\infty(B_2)$ by the appropriate domain. In particular,
 	$E^\pm$ are uniquely determined by $P$ and are continuous.
 \end{Theorem}	
 In $(iii)$ the space of distributional sections $\DD'(B)$ of a bundle $B$ over $M$ is defined
 as the topological dual of $\Gamma_0^\infty(B^*\otimes\Omega)$, where $\Omega$ is the bundle of weight-$1$ densities over $M$. As sections of $B$ and $B^*\otimes\Omega$ can be paired to give a density, there is an obvious embedding of $\Gamma^\infty(B)$ in $\DD'(B)$, with $\phi\in\Gamma^\infty(B)$ corresponding to the distribution
 $f\mapsto \int_M \langle \phi,f\rangle$ acting on $f\in \Gamma_0^\infty(B^*\otimes\Omega)$.
 The map $P$ in $(iii)$ is the restriction of the dual map
 $\big(\mu(\adj{P})\mu^{-1}\big)'$ to elements of $\DD'(B_1)$ with past- or future-compact support.
 (Recall that the formal dual is defined relative to the specific density~$\mu$.)
 In~\cite{Baer:2015}, B\"ar defines distributional sections of $B$ using
 the topological dual of $\Gamma_0^\infty(B^*)$, which would be the space of distribution densities in our terminology. The metric density
 provides an isomorphism between the spaces and operators that he considers and those that we do. As we have in mind potential applications where more than one metric might be in use, we have elected not to make a fixed identification between distributions and distribution densities.

 The \emph{second main theorem of Green-hyperbolicity} provides two exact sequences that are highly useful in applications, combining (and very mildly extending)\footnote{The extension is the right-most arrow, asserting surjectivity of $P$ onto spacelike compact (distributional) sections. But any such smooth section can be split as $f=f^++f^-$ where $f^\pm$ has past/future-compact support, whereupon $f=P\big( \overline{E}^+f^++\overline{E}^-f^-\big)$; the argument is identical for distributions, replacing $\overline{E}^\pm$ by $\widehat{E}^\pm$.} Theorems~3.22 and~4.3 of~\cite{Baer:2015}.
 \begin{Theorem} \label{thm:exact}
 	Let $P\colon\Gamma^\infty(B_1)\to \Gamma^\infty(B_2)$ be Green-hyperbolic and define the advanced-minus-retarded operators
 	$E=E^--E^+\colon\Gamma^\infty_0(B_2)\to\Gamma^\infty_{\mathrm{sc}}(B_1)$ and $\widehat{E}=\widehat{E}^--\widehat{E}^+$, using the notation of Theorem~$\ref{thm:GHbasics}$. Then there are two exact sequences
 	\begin{equation*}
 		\begin{tikzcd}
 			0 \arrow[r] & \GoinX{B_1} \arrow[r, "P"]\arrow[d] &\GoinX{B_2} \arrow[r, "E"]\arrow[d] & \Gamma^\infty_{\mathrm{sc}}(B_1) \arrow[r,"P"]\arrow[d] & \Gamma^\infty_{\mathrm{sc}}(B_2) \arrow[r]\arrow[d] & 0 \\
 			0 \arrow[r] & \DD'_0(B_1) \arrow[r, "P"] & \DD'_0(B_2) \arrow[r, "\widehat{E}"] & \DD'_{\mathrm{sc}}(B_1) \arrow[r,"P"] & \DD'_{\mathrm{sc}}(B_2) \arrow[r] & 0,
 		\end{tikzcd}
 	\end{equation*}
 	in which the downward arrows are the natural embeddings of smooth into distributional sections.
 \end{Theorem}

 A direct consequence is that the solution space $\Sol(P)=\{\phi\in\Gamma^\infty_{sc}(B)\colon P\phi=0\}$ is given by
 $\Sol(P)=E\GoinX{B}$.
 The special properties of Green-hyperbolic operators are not confined to the statement of Theorems~\ref{thm:GHbasics} and~\ref{thm:exact}. As shown in~\cite{Baer:2015}, products and direct sums of Green-hyperbolic operators are Green-hyperbolic, and indeed any operator whose square is Green-hyperbolic is Green-hyperbolic.
 Turning to physical applications, suppose that a bundle $B$ admits a nondegenerate bilinear form, or equivalently a base-point preserving vector bundle isomorphism $\mathcal{I}\colon B\to B^*$, and an antilinear base-point preserving involution $\mathcal{C}\colon B\to B$. Then $P\colon \Gamma(B)\to\Gamma(B)$ is said to be \emph{formally self-adjoint} if $P= \mathcal{I}^{-1}\adj{P}\mathcal{I}$, and real if $P=\mathcal{C} P \mathcal{C}$.
 If a~formally self-adjoint operator $P$ admits advanced and retarded Green operators then $P$ is Green-hyperbolic and moreover $\Sol(P)$ admits a symplectic form
 \begin{equation*}
 	\sigma(Ef_1,Ef_2) = \langle \mathcal{I} f_1,Ef_2\rangle, \qquad f_i\in\GoinX{B} .
 \end{equation*}
 If $P$ is also real then there is an associated bosonic QFT described by a unital $*$-algebra of observables, generated by symbols $\Phi(f)$, $f\in\Gamma_0^\infty(B)$, and subject to the relations:
 \begin{itemize}\itemsep=0pt
 	\item $f\mapsto\Phi(f)$ is complex linear (linearity),
 	\item $\Phi(f)^* = \Phi(\mathcal{C} f)$ (Hermiticity),
 	\item $\Phi(Pf)=0$ (field equation),
 	\item $[\Phi(f),\Phi(h)]={\rm i}\sigma(Ef,E h) \II$ (canonical commutation relations)
 \end{itemize}
 for all $f,h\in \Gamma_0^\infty(B)$. This essentially functorial quantisation is one of the main applications for the theory of Green-hyperbolic operators. Under more restrictive circumstances first order Green-hyperbolic operators can also admit fermionic quantisation~\cite{BaerGinoux:2012}.

 {\bf Modified Green-hyperbolic operators.}
 Turning to the subject of this paper, suppose $P\colon \Gamma^\infty(B_1)\to \Gamma^\infty(B_2)$ is Green-hyperbolic, with Green operators $E^\pm$. Let $A\colon\Gamma^\infty(B_1)\to \Gamma^\infty(B_2)$ be linear, with range contained in $\Gamma^\infty_K(B_2)$ for some compact, topologically regular $K\subset M$.

 Although $P+A$ is not necessarily a differential operator, we wish to state conditions analogous to (G1)--(G3) that can characterise suitable Green operators.
 To gain some insight, let us assume for the moment that
 for each $h\in \GoinX{B_2}$, the equation $(P+A)\phi=h$ has unique solutions with past/future-compact support given by $\phi=E_{P+A}^\pm h$ where $E_{P+A}^{\pm}\colon \GoinX{B_2}\to \Gamma^\infty(B_1)$ are linear maps with ranges
 necessarily contained in $\Gamma^\infty_{pc/fc}(B_1)$. Then we may write
 \begin{equation*}
 	PE_{P+A}^\pm f = f- AE_{P+A}^\pm f = PE_P^\pm\big(f- AE_{P+A}^\pm f\big),\qquad f\in \GoinX{B_2},
 \end{equation*}
 using our assumptions on $P$ and $A$. As $P$ is invertible on $\Gamma^\infty_{\pc/\fc}(B_1)$ by Theorem~\ref{thm:GHbasics}\,$(ii)$, we have
 \begin{equation*}
 	E_{P+A}^\pm f = E_P^\pm \big(f-AE_{P+A}^\pm f\big)
 \end{equation*}
 with support determined using condition (G3) for $P$,
 \begin{equation*}
 	\supp E_{P+A}^\pm f \subset J^\pm \big(\supp \big(f-AE_{P+A}^\pm f\big)\big) \subset J^\pm(\supp f\cup K).
 \end{equation*}
 By imposing further conditions on $A$ we may refine the information available. Specifically, suppose that $\phi|_K\equiv 0$ implies that $A\phi\equiv 0$.
 Observe that
 \begin{equation*}
 	(P+A)E_P^\pm f = f+ AE_P^\pm f \in \GoinX{B_1}
 \end{equation*}
 by (G2) for $E_P^\pm$ and the definition of $A$; by assumption on $E^\pm_{P+A}$ we now have
 \begin{equation*}
 	E_P^\pm f = E_{P+A}^\pm \big(f+ AE_P^\pm f\big)
 \end{equation*}
 and deduce that $E_{P+A}^\pm f=E_P^\pm f$ for all $f\in \GoinX{B_2}$ such that
 $J^\pm(\supp f)\cap K$ is empty \big(because~$AE_P^\pm f$ vanishes\big). Together with our earlier observation, $E^\pm_{P+A}$ satisfy the modified support property
 \begin{itemize}\setlength{\leftskip}{0.10cm}
\item[(G3${}'$)] for all $f\in \GoinX{B_2}$,
 	\begin{equation*}\displaystyle \supp E_{P+A}^{\pm} f \subset
 		\begin{cases} J^{\pm}(\supp f), & J^\pm(\supp f)\cap K=\varnothing,\\
 			J^{\pm}(\supp f\cup K), & \text{otherwise.}
 	\end{cases}\end{equation*}
 \end{itemize}
 With this intuition established, we can drop the assumption that $(P+A)\phi=h$ has unique solutions with past/future-compact support. The standing assumptions are now
 \begin{itemize}
 	\item[(A1)] $P\colon\Gamma^\infty(B_1)\to \Gamma^\infty(B_2)$ is Green-hyperbolic,
 	\item[(A2)] $A\colon\Gamma^\infty(B_1)\to \Gamma^\infty(B_2)$ is linear, with range contained in $\Gamma^\infty_K(B_2)$,
 	\item[(A3)] for $\phi\in \Gamma^\infty(B_1)$, $\phi|_K\equiv 0$ implies $A\phi\equiv 0$.
 \end{itemize}
 We now make the following definition.
 \begin{Definition}\label{def:KnonlocalGops}
 	Subject to assumptions (A1)--(A3), linear maps $E_{P+A}^{\pm}\colon\GoinX{B_2}\to \Gamma^\infty(B_1)$ are said to be \emph{retarded/advanced $K$-nonlocal Green operators} for $P+A$ if they satisfy~(G1) and~(G2) \big(with $P+A$ replacing $P$ and $E_{P+A}^\pm$ replacing $E^\pm$\big) and (G3${}')$. If both $P+A$ and $\adj{P}+\adj{A}$ admit retarded and advanced $K$-nonlocal Green operators then $P+A$ is called \emph{$K$-nonlocally Green-hyperbolic}.
 \end{Definition}

 As with Green-hyperbolic operators, the above definition implies considerably more and indeed the analogue of Theorem~\ref{thm:GHbasics} holds with (G3) replaced by (G3${}')$. These results will be proved in~\cite{FewVer:20xx}. Our main goal here will be to give sufficient
 conditions for the existence of $K$-nonlocal Green operators. For simplicity of presentation, we restrict to operators acting on spaces of smooth functions, rather than bundle sections. We will also establish some continuity results for the $K$-nonlocal Green operators, both from $C_0^\infty(M)$ to $C^\infty(M)$ and between various Sobolev spaces, and address the
 holomorphicity of the $K$-nonlocal Green operators with respect to a~parameter.

 At a formal level it is straightforward to see what the Green operators of $P+A$ should be, if they exist: if
 $\phi\in C^\infty_{\pc/\fc}(M)$ solves $(P+A)\phi=f\in \CoinX{M}$, then
 \begin{equation*}
 	\phi = E_P^\pm (f- A\phi) = E^\pm_{P} g,
 \end{equation*}
 where $g=f - A \phi= f- AE_P^\pm g$, so \smash{$g=\big(I+AE_P^\pm\big)^{-1}f$}. Formally, therefore, the Green operators for $P+A$ are
 \begin{equation}\label{eq:formalGF}
 	E_{P+A}^\pm = E_P^\pm \big(I+AE^\pm_{P}\big)^{-1},
 \end{equation}
 and the technical task is to make this formula rigorous, where possible, and to establish that the resulting operators $E_{P+A}^\pm$ are indeed $K$-nonlocal Green operators for $P+A$. Note that the inversion of $I+AE^\pm$ must be performed in $\CoinX{M}$. We will accomplish this by first inverting in Sobolev spaces, where Hilbert space techniques can be used, and then boosting the result up to $\CoinX{M}$.

 As can be seen from~\eqref{eq:formalGF}, the operator $\big(I+AE_P^\pm\big)^{-1}$ acts as a bridge between the free dynamics represented by $P$ and the interacting dynamics of $P+A$. This link can be expressed precisely in terms of classical M{\o}ller operators -- see Corollary~\ref{cor:Moller}.

 \section{Main result and remarks}\label{sec:main}

 Let $M$ be a globally hyperbolic spacetime with at most finitely many connected components,
 and let $K$ be a fixed, topologically regular, compact subset of $M$. For brevity, we write $C^\infty$ for~$C^\infty(M)$ and so on; for a closed set $A$, $C^\infty_A$ denotes the space of smooth functions with support in $A$, and the same convention is used for other spaces of distributions and Sobolev spaces. If $X$ and $Y$ are locally convex topological spaces, $\LL_b(X,Y)$ denotes the space of linear maps between them, equipped with the topology of bounded convergence (see Appendix~\ref{appx:topspaces}); for linear maps between normed spaces, this coincides with the usual operator norm topology. As usual, we write $\LL_b(X)$ for $\LL_b(X,X)$.

 The main result of this paper is the following.
 \begin{Theorem} \label{thm:GHpert}
 	Let $P$ be a Green-hyperbolic operator and
 	suppose $A(\lambda)$, $\lambda\in\CC$, is a family of linear maps $A(\lambda)\colon C^\infty\to C^\infty$ whose ranges are contained in $C^\infty_K$ and with $A(0)=0$. Suppose that,
 	for some $\beta,\gamma\in \RR$ with $\delta=\beta+\gamma>0$, and some $s_*\in \RR$,
 	\begin{enumerate}\itemsep=0pt
 		\item[$(a)$] the Green operators $E^\pm$ of $P$ extend to linear maps
 		$\EE'\to \DD'$ with continuous restrictions mapping $H^s_0\to H^{s+\beta}_\loc$ for all $s\ge s_*$;
 		\item[$(b)$] each map $A(\lambda)$ extends to a linear map $\DD'\to\DD'_K$ with continuous restrictions mapping $H^s_\loc\to H^{s+\gamma}_K$ for all $s\ge s_*+\beta$
 		$($consequently, the continuous maps $A(\lambda)E^\pm\colon H_0^s\to H_K^{s+\delta}$ induce
 		compact maps $A(\lambda)E^\pm\colon H_0^s\to H_K^{s}$ due to the Sobolev embedding theorems and $\delta>0)$;
 		\item[$(c)$] for all $s\ge s_*$, the compact map $H^s_0\to H^s_K$ induced by the compositions $A(\lambda)E^\pm$ is holomorphic in $\lambda\in\CC$ with respect to the topology of $\LL_b(H^s_0,H^s_K)$;
 		\item[$(d)$] if $f\in C^\infty$ vanishes identically on $K$ then $Af=0$.
 	\end{enumerate}
 	Then
 	\begin{enumerate}\renewcommand{\labelenumi}{(\Alph{enumi}).}
 		\item[$(A)$] $\ker (P+ A(\lambda))|_{C^\infty_{\pc/\fc}}$ has finite dimension
 		for each $\lambda\in\CC$;
 		\item[$(B)$]
 the sets
 		\begin{equation*}
 			S^\pm = \big\{\lambda\in \CC\colon \ker (P+ A(\lambda))|_{C^\infty_{\pc/\fc}}\neq 0\big\} \qquad \text{and}\qquad
 			S = S^+\cup S^-
 		\end{equation*}
 		are discrete subsets of $\CC$, whose complements are open neighbourhoods of zero;
 		\item[$(C)$]
 $P+A(\lambda)$ has advanced and retarded $K$-nonlocal Green operators for all $\lambda\in\CC\setminus S$, that are holomorphic in $\lambda$ on this domain with respect to the topology of $\LL_b(C_0^\infty, C^\infty)$; these Green operators are given explicitly as
 		\begin{equation*}
 			\tilde{E}_{\lambda}^\pm = E^\pm \big(I+A(\lambda)E^\pm\big)^{-1},
 		\end{equation*}
 		where the inverse on the right-hand side exists and is holomorphic in $\lambda\in \CC\setminus S$
 		with respect to the topology of $\LL_b(C_0^\infty)$;
 		\item[$(D)$]
 the Green operators
 		possess continuous extensions mapping $H^s_0\to H^{s+\beta}_\loc$ for all $s\ge s_*$ and which are holomorphic in $\lambda\in\CC\setminus S$ with respect to the topology of $\LL_b\big(H^s_0,H^{s+\beta}_\loc\big)$;
 		
 		\item[$(E)$]
 if $(d)$ is replaced by
 		\begin{enumerate}
 			\item[$(d')$] $\supp A(\lambda)f\subset \supp f$ for all $f\in C^\infty$, and $\lambda\in\CC$
 		\end{enumerate}
 		then $P+ A(\lambda)$ has advanced and retarded Green operators in the usual sense
 		for $\lambda\in\CC\setminus S$.
 	\end{enumerate}
 \end{Theorem}
 Of course, in the situation where the hypotheses also apply to the formal duals of $P$ and $A(\lambda)$ then the upshot is that $P+ A(\lambda)$ and $\adj{P}+\adj{A}(\lambda)$
 both admit advanced and retarded $K$-nonlocal Green operators for all
 $\lambda\in\CC$ in an open $0$-neighbourhood with discrete complement given by the union of $S$ with the corresponding set for $\adj{P}$. For the non-exceptional values, $P+A(\lambda)$ is $K$-nonlocally Green-hyperbolic. The general theory of such operators and their properties is developed in detail in~\cite{FewVer:20xx}. We note in particular that when
 the ranges of both $A(\lambda)$ and $\adj{A}(\lambda)$ are contained in $C^\infty_K$ then assumption $(d)$ holds automatically for both operators.

 In many situations the holomorphicity assumption $(c)$ is easily verified, e.g., where $A(\lambda)$ is a~polynomial in $\lambda$. Precomposing with the continuous embedding $H^s_K\to H^s_0$, it follows that the restrictions of $A(\lambda)E^\pm$ to $H^s_K$, for $s\ge s_*$, are holomorphic with respect to the norm topology
 of $\LL(H^s_K)$. On the other hand, postcomposing with the embedding $H^s_K\to H^s_0$, we also have holomorphicity of $A(\lambda)E^\pm$ with respect to
 $\LL_b(H^s_0)$ (e.g., see Lemma~\ref{lem:Lb_basic}).

 An immediate application of Theorem~\ref{thm:GHpert} is:
 \begin{Corollary}\label{cor:Moller}	
 	Let $P$ be a second-order normally hyperbolic operator on $M$ and let $K\subset M$ be compact.
 	Let $\rho\colon\CC\to C^\infty(K\times K)$ be a polynomial and define $A(\lambda)\colon C^\infty(M)\to C_K^\infty(M)$ by
 	\begin{equation*}
 		(A(\lambda) \phi)(x)=\int_M \rho(\lambda)(x,y)\phi(y)\mu_y, \qquad \phi\in C^\infty(M).
 	\end{equation*}
 	Then $P$ and $A(\lambda)$ satisfy the hypotheses of Theorem~$\ref{thm:GHpert}$. In particular, $P+\lambda A_\rho$ has advanced and retarded $K$-nonlocal Green operators for all $\lambda$ in an open $0$-neighbourhood with a discrete complement in $\CC$.
 \end{Corollary}
 \begin{proof}
 	$P$ is Green-hyperbolic, and its
 	Green operators extend to maps $E^\pm\colon H^s_0\to H^{s+1}_\loc$, i.e., they improve smoothness by one order, by~\cite[Theorem 6.5.3]{DuiHoer_FIOii:1972} (see~\cite{IslamStrohmaier:2020} for a bundle version).
 	One easily checks that the technical conditions on $A(\lambda)$ are certainly met by integral operators with kernels in $C^\infty_{K\times K}(M\times M)$.
 \end{proof}

 Another consequence is the existence of classical M{\o}ller operators~\cite{DuetschFredenhagen:2003,HawRej:2020} relating the solution spaces of $P$ and $P+A(\lambda)$.
 \begin{Corollary}
 	Under the hypotheses of Theorem~$\ref{thm:GHpert}$,
 	let $\Sf(\lambda)=\{\phi\in C^\infty\colon (P+A(\lambda))\phi=0\}$. Then the maps
 	$\mathsf{r}_\lambda:C^\infty\to C^\infty$ given by
 	\begin{equation}\label{eq:Moller}
 		\mathsf{r}_\lambda \phi = \phi - E^+\big(I+A(\lambda)E^+\big)^{-1}A(\lambda)\phi,
 		\qquad \phi\in C^\infty,
 	\end{equation}
 	are holomorphic in $\lambda$ on $\CC\setminus S$ with respect to the topology of $\LL_b(C^\infty)$. For each $\lambda\in \CC\setminus S$, $\mathsf{r}_\lambda$ is a homeomorphism of $C^\infty$ that restricts to a bijection from $\Sf(0)$ to $\Sf(\lambda)$.
 \end{Corollary}
 Because $\mathsf{r}_\lambda \phi$ and $\phi$ differ only in $J^+(K)$, the causal future of the interaction region, one calls~$\mathsf{r}_\lambda$ the \emph{retarded classical M{\o}ller operator} intertwining the dynamics of $P$ and $P+A(\lambda)$. An advanced M{\o}ller operator can be constructed by replacing $E^+$ by $E^-$ in~\eqref{eq:Moller}.
 \begin{proof}
 	By part~$(C)$ of Theorem~\ref{thm:GHpert}, it is clear that $\mathsf{r}_\lambda\in \LL(C^\infty)$ and holomorphic in $\lambda$ on $\CC\setminus S$ with respect to the $\LL_b(C^\infty)$ topology.
 	A direct calculation, using invertibility of $I+A(\lambda)E^+$ on~$C_0^\infty$, shows that $\mathsf{r}_\lambda$ is invertible on~$C^\infty$ with continuous inverse given by
 	\begin{equation*}
 		\mathsf{r}_\lambda^{-1}\phi = \phi + E^+A(\lambda)\phi, \qquad \phi\in C^\infty.
 	\end{equation*}
 	Finally, if $\phi\in \Sf(0)$, then
 	\begin{equation*}
 		(P+A(\lambda))\mathsf{r}_\lambda\phi = A(\lambda)\phi -
 		\big(I+A(\lambda)E^+\big)\big(I+A(\lambda)E^+\big)^{-1}A(\lambda)\phi =0
 	\end{equation*}
 	and, if $\phi\in \Sf(\lambda)$, then
 	\begin{equation*}
 		P\mathsf{r}_\lambda^{-1}\phi= -A(\lambda)\phi + A(\lambda)\phi= 0
 	\end{equation*}
 	so $\mathsf{r}_\lambda\Sf(0)\subset \Sf(\lambda)$ and $\mathsf{r}_\lambda^{-1}\Sf(\lambda)\subset \Sf(0)$, from which we deduce that $\mathsf{r}_\lambda$ restricts to a bijection from~$\Sf(0)$ to $\Sf(\lambda)$.
 \end{proof}

 Theorem~\ref{thm:GHpert} is proved in the next section, after a number of further remarks and examples. First, we note that Theorem~\ref{thm:GHpert} is far from the most general statement that could be made. For instance, at the cost of more notation but no new ideas, it generalises immediately to operators acting on sections in finite-dimensional vector bundles over $M$. This can also be shown by applying Theorem~\ref{thm:GHpert} in its current form, using a method of LU factorisation described in Section~\ref{sec:LU}.

 Second, generalisations to multi-variable modifications $P+ \sum_{j=1}^r A_j(\lambda_j)$ are possible. In this case, the discrete exceptional sets are replaced by sets in $\CC^r$ that, locally, are vanishing sets of holomorphic functions in $r$ variables. See~\cite{Gramsch:1970} for the appropriate multi-variable Fredholm theorem and~\cite{Kaballo:2012} for an exposition.

 Third, the significance of holomorphicity in the topology of bounded convergence is that compositions of such functions are also holomorphic in the topology of bounded convergence, with derivatives given by the usual Leibniz rule. This is explained in more detail in Appendix~\ref{appx:top_bded_cvgnce}.

 Fourth, Theorem~\ref{thm:GHpert} shows that the obstruction to the existence of $K$-nonlocal Green functions for $P+\lambda A$ is provided by nontrivial smooth solutions to $(P+\lambda A)\phi=0$ with past- or future-compact support (which, if present, span finite-dimensional spaces). Any sort of energy estimate would be sufficient to exclude nontrivial solutions of this type, which indicate unphysical behaviour in a closed system. Mathematical examples are easily constructed, however.
 \begin{Example}
 	Let $P$ be Green-hyperbolic on $M$ with Green operators $E^\pm$. Fix any nontrivial $f,h\in C^\infty_K$ and
 	define $A\colon C^\infty\to C_0^\infty$ by
 	\begin{equation*}
 		A \phi =- \left(\int h \phi\, \mu \right)f,
 	\end{equation*}
 	noting that the range of $A$ is contained in $C^\infty_K$.
 	If $\phi\in C^\infty_{\pc/\fc}$ obeys $(P+\lambda A)\phi=0$, then $P\phi = -(\lambda \int h\phi\,\mu)f$ and
 	we deduce that $\phi$ is a constant multiple of $E^\pm f$; hence it must also be that $(P+\lambda A)E^\pm f=0$ or equivalently $f = \lambda \nu^\pm f$ where
 	$\nu^\pm = \int h E^\pm f\mu$. Thus a necessary condition for the existence of nontrivial $\phi\in\ker (P+\lambda A)|_{C^\infty_{\pc/\fc}}$ is that $\lambda \nu^\pm=1$, and it is easily seen that this condition is sufficient. Summarising,
 	\begin{equation*}
 		\ker(P+\lambda A)|_{C^\infty_{\pc/\fc}}= \begin{cases} 0, & \lambda\nu^{\pm}\neq 1, \\
 			\CC E^{\pm}f, & 	\lambda\nu^{\pm}= 1.
 		\end{cases}
 	\end{equation*}
 	Next, note that $\adj{A}$ takes the same form as $A$ but with $f$ and $h$ exchanged. Therefore,
 	$\ker \adj{(P+\lambda A)}|_{C^\infty_{\pc/\fc}}$ is nontrivial if and only if
 	$\lambda \adj{\nu}^\pm = 1$, where $\adj{\nu}^\pm= \int f E_t^\pm h\mu$ and $E_t^\pm$ are the Green operators of $\adj{P}$. But $E_t^\pm$ is the formal dual of $E^\mp$, so in fact $\adj{\nu}^\pm = \nu^\mp$ and
 	\begin{equation*}
 		\ker\adj{(P+\lambda A)}|_{C^\infty_{\pc/\fc}}= \begin{cases} 0, & \nu^{\mp}\lambda\neq 1, \\
 			\CC E^{\pm}h, & 	 \nu^{\mp}\lambda= 1
 		\end{cases}
 	\end{equation*}
 	and consequently
 	\begin{equation*}
 		\dim\ker (P+ \lambda A)|_{C^\infty_{\pc/\fc}} = \dim\ker \big(\adj{P} +\lambda \adj{A}\big)\big|_{C^\infty_{\fc/\pc}} .
 	\end{equation*}
 	Thus there are at most two values of $\lambda$ for which $P+\lambda A$ can fail to be $K$-nonlocal Green-hyperbolic.
 	In the case where $P$ obeys the hypothesis $(a)$ of Theorem~\ref{thm:GHpert} (e.g., if $P$ is the Klein--Gordon operator) then $P+\lambda A$ is $K$-nonlocally Green-hyperbolic for all $\lambda\in\CC\setminus \big\{(\nu^\pm)^{-1}\big\}$.
 \end{Example}	

 The example above illustrates a general result based on Fredholm index theory, which was prompted by an insightful
 question posed to the author by B\"ar, and is proved in Section~\ref{sec:Fredholm}.
 \begin{Theorem}\label{thm:Fredholm}
 	Suppose that the hypotheses of Theorem~$\ref{thm:GHpert}$ are met for both $P$ and $A(\lambda)$, and~$\adj{P}$ and $\adj{A}(\lambda)$, with $\beta\ge 0$ and $s_*\le -\beta<\gamma$.
 	Then, for all $\lambda\in\CC$,
 	\begin{equation*}
 		N^\pm(\lambda) := \dim\ker (P+A(\lambda))|_{C^\infty_{\pc/\fc}} = \dim\ker \big(\adj{P}+\adj{A}(\lambda)\big)\big|_{C^\infty_{\fc/\pc}} .
 	\end{equation*}
 	In particular, $N^+(\lambda)=N^-(\lambda)$ in the case where $P$ and $A(\lambda)$ are formally self-dual.
 \end{Theorem}
 At first sight it is quite surprising that the spaces of spontaneously appearing and disappearing solutions have equal dimension for the self-dual case. At its core is the fundamental fact about
 Lorentzian causality that $x\in J^\pm(y)$ if and only if $y\in J^\mp(x)$, and its consequence that
 the advanced and retarded Green operators of a Green-hyperbolic operator $P$ are formal duals of the retarded and advanced Green operators of the formal dual $\adj{P}$. This by itself is not enough for the result above, which also makes essential use of the that $A(\lambda)E^\pm$ is compact, because it improves regularity.

 The last remark prompts one to consider situations where $A(\lambda)E^\pm$ does not improve regularity.
 \begin{Example}\label{ex:nihilo}
 	Consider the Green-hyperbolic operator $P=\partial_u\partial_v$ with respect to $(u,v)$ coordinates on $\RR^2$, regarding vectors with nonnegative components in these coordinates as future-pointing and causal. We may write $P=\partial\otimes\partial$ in an obvious tensor product notation. Let~$f$,~$g$ and~$h$ be smooth real-valued functions on $\RR$, so that $f$, $g'$ and $h$ have support contained in $[-2,2]$, $\supp g\subset [-2,\infty)$, $f\equiv 1$ on a neighbourhood of $[-1,1]$ and $\ip{h}{g}=1$ in the usual~$L^2(\RR)$ inner product.
 	Define $T=-\ket{g'}\bra{h}$ and note that
 	\begin{equation*}
 		g' + T g = 0.
 	\end{equation*}
 	Setting $K=[-2,2]\times[-2,2]$, the operator
 	\begin{equation*}
 		A = (f\partial)\otimes T
 	\end{equation*}
 	maps $C^\infty\big(\RR^2\big)$ continuously to $\CoinX{K}$, vanishes on
 	$\phi\in C^\infty(M)$ with $\phi|_K\equiv 0$, and extends to a continuous map $H^s_\loc\to H^{s-1}_K$ for any $s\in\RR$. Thus $AE^\pm\colon H^s_K\to H^s_K$. It is now easily seen that the equation
 	\begin{equation*}
 		P\varphi + A \varphi = 0
 	\end{equation*}
 	is solved by any distribution $\varphi=\upsilon\otimes g$, where $\upsilon\in\DD'(\RR)$ is supported in $[-1,1]$. Such $\varphi$ have past compact support, contained in $[-1,1]\times [-2,\infty)$.
 	
 	This example shows that the regularity-improving nature of $AE^\pm$ is
 	responsible for both the smoothness of past-/future-compact solutions to $(P+A)\phi=0$ and the finite-dimensionality of the corresponding solution spaces.
 \end{Example}

 In the light of this example it is clear that further conditions would be needed to deal with modifications of first order derivative operators. For example, if $D$ is a Dirac operator then
 there is a companion operator $\tilde{D}$ so that $P = D\tilde{D}$
 and $\tilde{P}=\tilde{D}D$ are second order Green-hyperbolic with Green operators $E^\pm$ and $\tilde{E}^\pm$ respectively that improve regularity by one order. Then $G^\pm = \tilde{D}E^\pm$ and $\tilde{G}^\pm = D\tilde{E}^\pm$ are Green operators for $D$ and $\tilde{D}$ respectively. Similarly,
 $G_t^\pm = \adj{\tilde{D}}E_t^\pm$ and $\tilde{G}_t^\pm = \adj{D}\tilde{E}_t^\pm$ are Green operators for $\adj{D}$ and $\adj{\tilde{D}}$, where $E_t^\pm$ and $\tilde{E}_t^\pm$ denote the Green operators for $\adj{P}$ and $\adj{\tilde{P}}$. Thus $D$ and $\tilde{D}$ are Green-hyperbolic, but we cannot assume
 that their Green operators improve regularity.

 Now consider a $K$-nonlocal modification $D+ A$ of $D$, suppressing $\lambda$-dependence in the following. If $A$ and $\adj{A}$ have range in $C^\infty_K$ and \emph{improve} regularity, then Theorem~\ref{thm:GHpert} applies and shows that $D+ A$ is $K$-nonlocally Green-hyperbolic provided that there are no nontrivial past-compact or future-compact solutions to $(D+A)\phi=0$ or $\big(\adj{D}+\adj{A}\big)\phi=0$. More generally, a natural strategy to find Green operators for the modified operator is to seek a $K$-nonlocal operator $\tilde{A}$ so that
 \begin{equation*}
 	P_A = (D+A)\big(\tilde{D}-\tilde{A}\big) = D\tilde{D} +A\tilde{D}-D \tilde{A} - A\tilde{A}
 \end{equation*}
 and
 \begin{equation*}
 	\tilde{P}_A = \big(\tilde{D}-\tilde{A}\big)(D+A) = \tilde{D}D +\tilde{D}A -\tilde{A}D - \tilde{A}A
 \end{equation*}
 are $K$-nonlocally Green-hyperbolic. If $A$ is regularity-preserving and one may find a
 regularity-preserving $\tilde{A}$ such that $A\tilde{D}-D\tilde{A}$ and $\tilde{D}A -\tilde{A}D$ lose \emph{strictly less} than one order of regularity (for all $\lambda\in\CC$), then Theorem~\ref{thm:GHpert} could be applied to $P_A$ and $\tilde{P}_A$ and their Green operators (when they exist) may be used to construct Green operators for $D+A$ and $\tilde{D}+\tilde{A}$ as before. Lacking such a choice, one potentially falls victim to the behaviour in Example~\ref{ex:nihilo}, unless Theorem~\ref{thm:GHpert} can be extended to situations in which $AE^\pm$ does not improve regularity.
 As these questions are best pursued in the context of specific models, we leave the discussion here
 and turn to the proof of Theorem~\ref{thm:GHpert}.

 \section{Application: LU factorisation method for systems}\label{sec:LU}

 LU factorisation is a standard technique in linear algebra for solving system of linear equations $Mx=y$ for $x,y\in \CC^N$. In situations where an invertible matrix $M=LU$ with $L$ lower-triangular and $U$ upper-triangular, the solution is given by solving the triangular systems $Lz=y$, $Ux=z$. The method can be generalised in many ways, see, e.g.,~\cite{HornJohnson_vol1}. Here, we discuss its use to solve systems of Green-hyperbolic equations, or $K$-nonlocal generalisations thereof.
 To see why this cannot be done straightforwardly, consider a system $\Ps \Phi = F$ where $\Ps\colon C^\infty\big(M;\CC^N\big)\to C^\infty\big(M;\CC^N\big)$ is a $N\times N$ matrix of
 linear self-maps of $C^\infty(M)$, where $M$ is a~globally hyperbolic spacetime as usual and $N\ge 2$. We adopt a block form
 \begin{equation*}
 	\Ps = \begin{pmatrix} P & R \\ S & Q\end{pmatrix},
 \end{equation*}
 where $Q$ is an $(N-1)\times (N-1)$ block, fixing the dimensions of the other blocks accordingly, so $P=\Ps_{11}$ is a self-map of $C^\infty(M)$. If $P$ is a differential operator with retarded and advanced Green operators $E^\pm_P$ and $R C^\infty\big(M;\CC^{N-1}\big)\subset
 \CoinX{M;\CC}$, one may factorise $\Ps$ on $\CoinX{M;\CC^N}$ as
 \begin{equation}\label{eq:LU1}
 	\Ps = \begin{pmatrix} 1 & 0 \\ SE_P^\pm & 1\end{pmatrix}\begin{pmatrix} P & R \\ 0 & Q - SE_P^\pm R \end{pmatrix}.
 \end{equation}
 This is the first step towards an LU factorisation: the first factor is indeed lower-triangular, but the second is not generally upper-triangular. The problem is that even if all the individual matrix elements of $\Ps$ are
 differential operators, the factorised form involves the typically nonlocal operators $SE^\pm_P$ and $SE^\pm_PR$. Therefore to proceed with this strategy one should rather phrase the problem from the start in terms of suitable nonlocal operators. (Modulo smoothing operators, one could employ pseudodifferential operators if the leading diagonal operators were elliptic, but here we require something more.)

 We describe how an exact LU factorisation can be achieved using Theorem~\ref{thm:GHpert}. Let $K$ be a~fixed compact, topologically regular subset of $M$, and let $s_*\in\RR$. Fix $\beta,\gamma\in\RR$ with $\beta+\gamma>0$. Let $\mathscr{A}$ be the space of maps $A\colon \CC\to \LL(C^\infty(M))$ so that, for all $\lambda\in\CC$,
 \begin{itemize}\itemsep=0pt
 	\item $A(\lambda)C^\infty(M)\subset C^\infty_K(M)$,
 	\item $A(\lambda)f\equiv 0$ if $f\in C^\infty(M)$ vanishes identically on $K$,
 	\item $A(\lambda)$ extends to a linear map $\DD'(M)\to\DD'_K(M)$ with continuous restrictions mapping $H^s_\loc\to H^{s+\gamma}_K$ for all $s\ge s_*+\beta$,
 \end{itemize}
 and so that $\lambda\to A(\lambda)$ is holomorphic with respect to the topology of $\LL_b\big(H^s_\loc,H_K^{s+\gamma}\big)$. Consider a system $\Ps$ depending on a parameter $\lambda$ so that
 every off-diagonal component $\Ps_{ij}$, $i\neq j$,
 of $\Ps$ is a map in $\mathscr{A}$, and the diagonal components take the form
 \begin{equation*}
 	\Ps_{ii} = P_i + A_i,
 \end{equation*}
 where each $P_i$ is a $\lambda$-independent Green-hyperbolic operator and each $A_i$ is a map in $\mathscr{A}$.
 Assume finally that the Green operators of the $P_i$ extend to maps $H^s_0\to H^{s+\beta}_\loc$ for all $s\ge s_*$.

 By Theorem~\ref{thm:GHpert}, $P(\lambda)=P_1+A_1(\lambda)$ has retarded and advanced $K$-nonlocal Green operators
 for all $\lambda$ in an open $0$-neighbourhood with discrete complement in $\CC$.
 We may therefore factorise~$\Ps$ as in~\eqref{eq:LU1}. A key point is that the $(N-1)\times (N-1)$-dimensional systems
 \begin{equation*}
 	\Ps^\pm(\lambda)= Q(\lambda)-S(\lambda)E_{P(\lambda)}^\pm R(\lambda)
 \end{equation*}
 obey the same assumptions as the original system, noting that the matrix components of $SE_P^\pm R$ determine continuous maps $H^s_0\to H_K^{s+\beta+\delta}\hookrightarrow H_K^{s+\beta}$. Theorem~\ref{thm:GHpert} implies that the leading diagonal component of $\Ps^+(\lambda)$ (resp.\ $\Ps^-(\lambda)$) has retarded (resp.\ advanced) $K$-nonlocal Green operators for all $\lambda$ outside a possibly enlarged exceptional set, so we may factor each of $\Ps^\pm(\lambda)$. Repeating the process, this leads to two LU factorisations of $\Ps(\lambda)$, which differ only in whether advanced or retarded Green operators are used in their construction. At each stage in the process we may gain more exceptional points, but the overall exceptional set is still discrete and excludes zero. From now on, we suppress the parameter $\lambda$ in the notation.

 For non-exceptional $\lambda$, we may now use the LU factorisation to obtain
 retarded and advanced Green operators for $\Ps$. We proceed inductively in $N$.
 When $N=1$ we are precisely in the situation of Theorem~\ref{thm:GHpert}.
 Now suppose that Green operators can be constructed for $(N-1)$-dimensional systems of the type considered, where $N\ge 2$. To establish the inductive
 step we return to the factorisation~\eqref{eq:LU1}, in which we may now
 assume that $Q-SE_P^\pm R$ has retarded/advanced Green operators and that $P$ has both retarded and advanced Green operators. We claim that
 \begin{equation}\label{eq:candidates}
 	E^\pm_{\Ps} = \begin{pmatrix}
 		E_P^\pm & -E_P^\pm RE_{Q-SE_P^\pm R}^\pm \\
 		0 & E_{Q-SE_P^\pm R}^\pm
 	\end{pmatrix}
 	\begin{pmatrix} 1 & 0 \\ -SE_P^\pm & 1\end{pmatrix}
 \end{equation}
 are retarded/advanced Green operators for $\Ps$. To check this, we first observe that the lower-triangular factor in~\eqref{eq:LU1}, which we denote $L$, is invertible on $\CoinX{M;\CC^N}$ with inverse
 \begin{equation*}
 	L^{-1} = \begin{pmatrix} 1 & 0 \\ -SE_P^\pm & 1\end{pmatrix}.
 \end{equation*}
 Then, we compute on the one hand, that
 \begin{equation*}
 	\Ps 	E^\pm_{\Ps} F = \begin{pmatrix} P & R \\ S & Q\end{pmatrix}
 	\begin{pmatrix}
 		E_P^\pm & -E_P^\pm RE_{Q-SE_P^\pm R}^\pm \\
 		0 & E_{Q-SE_P^\pm R}^\pm
 	\end{pmatrix}
 	L^{-1} F=
 	\begin{pmatrix} 1 & 0 \\ SE_P^\pm & 1\end{pmatrix}
 	L^{-1} F = F
 \end{equation*}
 and on the other, that
 \begin{equation*}
 	E^\pm_{\Ps}	\Ps F = \begin{pmatrix}
 		E_P^\pm & -E_P^\pm RE_{Q-SE_P^\pm R}^\pm \\
 		0 & E_{Q-SE_P^\pm R}^\pm
 	\end{pmatrix}
 	L^{-1}L \begin{pmatrix} P & R \\ 0 & Q - SE_P^\pm R \end{pmatrix} F=
 	F
 \end{equation*}
 for any $F\in \CoinX{M;\CC^N}$ in each case.
 Finally, the support property (G3${}'$) is clear, provided one takes a consistent choice of $+$ or $-$ in~\eqref{eq:candidates}. Thus $\Ps$ has retarded and advanced $K$-nonlocal Green operators, which is the inductive step, and shows that all finite-dimensional systems of the type considered possess advanced and retarded $K$-nonlocal Green operators; furthermore, these will vary holomorphically in $\lambda$, in the topology of $\LL_b\big(\CoinX{M;\CC^N},C^\infty\big(M;\CC^N\big)\big)$ outside the exceptional set. In particular, this justifies the treatment of such systems in a recent paper on measurement in QFT~\cite{FewsterJubbRuep:2022}, where interactions between a `system' QFT and one or more `probe QFT's are analysed.

 \section{Proof of Theorem~\ref{thm:GHpert}}\label{sec:proof}

 The proof of Theorem~\ref{thm:GHpert} starts by establishing an inversion result that will later be applied to $I+A(\lambda)E^\pm$, where $E^\pm$ are the Green operators of $P$. This is proved first in Sobolev spaces $H^s_K$, using the analytic Fredholm theorem, and then extended to $C^\infty_K$. The main part of the proof then uses this information to define Green operators for $P+ A(\lambda)$.

 \begin{Theorem}\label{thm:Sobolevinversion}
 	Let $K\subset M$ be a fixed topologically regular compact set. Suppose
 	$Y(\lambda)$, $\lambda\in\CC$, is a~family of linear maps
 	$Y(\lambda)\colon\DD_K'\to\DD'_K$, each restricting to continuous maps $Y_s(\lambda)\colon H^s_K\to H^{s+\delta}_K$ for all $s\ge s_*\in\RR$ and some fixed $\delta>0$. Assume also that $Y(0)=0$.
 	Let $\hat{Y}_s(\lambda)\in\LL(H^s_K)$ denote the compact maps obtained by composing $Y_s(\lambda)$ with the embedding $H^{s+\delta}_K\hookrightarrow H^s_K$. Suppose that $\lambda \to \hat{Y}_s(\lambda)$ is holomorphic on $\CC$ in $\LL_b(H^s_K)$ for all $s\ge s_*$. Then
 	\begin{enumerate}\itemsep=0pt
 		\item[$(a)$] The set
 		\begin{equation}\label{eq:Sdef}
 			S=\big\{\lambda\in\CC\colon \ker (I+ Y(\lambda))|_{C^\infty_K}\neq 0\big\}
 		\end{equation}
 		is a discrete subset of $\CC$, and $\CC\setminus S$ is an open $0$-neighbourhood. Furthermore, $\ker (I+ Y(\lambda))|_{C^\infty_K}$ is finite-dimensional, and equal to
 		$\ker (I+ \hat{Y}_s(\lambda))$ for each $s\ge s_*$.
 		
 		\item[$(b)$] For all $\lambda\in \CC\setminus S$ and all $s\ge s_*$, the map $I+ \hat{Y}_{s}(\lambda)$ is continuously invertible for all $s\ge s_*$, and $\lambda\mapsto \big(I+ \hat{Y}_{s}(\lambda)\big)^{-1}$ is holomorphic in $\CC\setminus S$. Moreover,
 		$\big(I+ \hat{Y}_{s}(\lambda)\big)^{-1}$ is the restriction of $\big(I+ \hat{Y}_{s_*}(\lambda)\big)^{-1}$ to $H^s_K$ for each $s\ge s_*$, $\lambda\in\CC\setminus S$.
 		
 		\item[$(c)$] $(i)$ Suppose that $f\colon \CC\setminus S\to H^{s_*}_K$ is holomorphic and
 		there is a compact subset ${K_f\subset K}$ so that $\supp \big(\hat{Y}_{s_*}(\lambda)\big)^r f(\lambda)\subset K_f$ for all $r\in\NN_0$ and $\lambda\in \CC\setminus S$. Then $\supp \big(I+\hat{Y}_{s_*}(\lambda)\big)^{-1}f(\lambda)\subset K_f$ for all $\lambda\in \CC\setminus S$.
 		
 		$(ii)$ If $Y(\lambda)C^\infty_K\subset C^\infty_{K'}$ for some compact $K'\subset K$, then
 		$\supp \big(I+ \hat{Y}_{s_*}(\lambda)\big)^{-1}f\subset K'$ for all $\lambda\in \CC\setminus S$ and all $f\in C^\infty_K$, and
 		$\ker \big(I+ \hat{Y}_{s_*}(\lambda)\big)$ is a finite-dimensional subspace of
 		$C^\infty_{K'}$ for all $\lambda\in\CC$.
 		
 		NB. In $(c)$, `support' is to be understood as distributional support.
 	\end{enumerate}
 \end{Theorem}
 \begin{proof}
 	$(a)$ Compactness of $\hat{Y}_s(\lambda)$ follows from the Sobolev embedding theorems.
 	The function $\lambda\mapsto I+ \hat{Y}_{s_*}(\lambda)$ is an analytic function on $\CC$ with values in the Fredholm operators on $H^{s_*}_K$, and which is invertible for $\lambda=0$ because $\hat{Y}_{s_*}(0)=0$. By the analytic Fredholm theorem~\cite[Theorem VI.14]{ReedSimon:vol1}, $I+ \hat{Y}_{s_*}(\lambda)$ is invertible for all $\lambda\in\CC$ with the exception of the (possibly empty) set $S_{s_*}$ of $\lambda\in \CC$ for which $\ker \big(I+ \hat{Y}_{s_*}(\lambda)\big)$ is nontrivial (and necessarily finite-dimensional by~\cite[Theorem VI.15]{ReedSimon:vol1}). Furthermore, $S_{s_*}$ is a
 	discrete subset of $\CC$, whose complement is an open $0$-neighbourhood, and the inverse $\big(I+ \hat{Y}_{s_*}(\lambda)\big)^{-1}$ is meromorphic on $\CC$ and holomorphic on $\CC\setminus S_{s_*}$.
 	
 	If $f\in\ker \big(I+\hat{Y}_{s_*}(\lambda)\big)$ then $f = - \hat{Y}_{s_*}(\lambda)f\in H^{s_*+\delta}_K$ obeys $f\in\ker \big(I+ \hat{Y}_{s_*+\delta}(\lambda)\big)$; iterating, $f\in \bigcap_{s\ge s_*} H^s_K=C^\infty_K$ and $f\in \ker (I+ Y(\lambda))|_{C^\infty_K}$.
 	As the converse inclusion is trivial, we deduce that
 	$\ker \big(I+\hat{Y}_{s_*}(\lambda)\big)= \ker (I+ Y(\lambda))|_{C^\infty_K}$
 	for all $\lambda\in \CC$ and hence that $S_{s_*}=S$ defined in~\eqref{eq:Sdef}. In particular, for each fixed $\lambda\in\CC$, all the kernels $\ker \big(I+\hat{Y}_{s}(\lambda)\big)$ are equal for $s\ge s_*$.
 	
 	$(b)$ The first statement is immediate by the Fredholm theorem and part~(a); the second follows because $I+\hat{Y}_{s}(\lambda)$ coincides with $I+ \hat{Y}_{s_*}(\lambda)$ on $H^s_K$ and both operators are invertible for $\lambda\in\CC\setminus S$.
 	
 	$(c)$\,$(i)$ Let $\chi\in C_0^\infty$ vanish on $K_f$. Regarding each $\big(\hat{Y}_{s_*}(\lambda)\big)^{n}f(\lambda)$ as a distribution with support in $K_f$, we see (e.g., by \cite[Theorem 2.3.3]{Hormander1}) that
 	$\big(\hat{Y}_{s_*}(\lambda)^{n}f(\lambda)\big)(\chi)=0$ for all $n\in \NN_0$ and so the
 	holomorphic function on $\CC\setminus S$ defined by $\lambda\mapsto \big(\big(I+\hat{Y}_{s_*}(\lambda)\big)^{-1}f(\lambda)\big)(\chi)$ vanishes in a neighbourhood of the origin on which the resolvent may be expanded as a convergent geometric series in powers of $\hat{Y}_{s_*}(\lambda)$, recalling that $\hat{Y}_{s_*}(\lambda)\to 0$ as $\lambda\to 0$.
 	By holomorphicity, it follows that
 	$\big(\big(I+\hat{Y}_{s_*}(\lambda)\big)^{-1}f(\lambda)\big)(\chi)=0$ for all $\lambda\in\CC\setminus S$.
 	Allowing $\chi$ to vary, we conclude that $\supp \big(I+\hat{Y}_{s_*}(\lambda)\big)^{-1}f(\lambda)\subset K_f$.
 	
 	$(c)$\,$(ii)$ The assumption implies that of $(c)$\,$(i)$ for all $f\in C^\infty_K$ \big(regarded as constant functions from $\CC\setminus S$ to $H^{s_*}_K$\big) with $K_f=K'$, so
 	$\supp \big(I+\hat{Y}_{s_*}(\lambda)\big)^{-1}f\subset K'$ for all $\lambda\in \CC\setminus S$ and all $f\in C^\infty_K$. Finally,
 	if $g\in \ker \big(I+\hat{Y}_{s_*}(\lambda)\big)=\ker (I+ Y(\lambda))|_{C^\infty_K}$ (already known to be finite-dimensional) then $g\in C^\infty_K$ and $g=- Y(\lambda) g\in C^\infty_{K'}$.
 \end{proof}

 The hypotheses of Theorem~\ref{thm:Sobolevinversion} entail that each $Y(\lambda)$ restricts to a continuous endomorphism of $C^\infty_K$. To see this, note that for any $k\in\NN_0$ and any integer $s>k+n/2$, where $n$ is the maximum dimension of any component of $M$, there are continuous maps
 \begin{equation*}
 	C^s_K \longrightarrow H^s_K \stackrel{Y(\lambda)}{\longrightarrow} H^s_K \longrightarrow C^k_K,
 \end{equation*}
 where the unlabelled arrows are Sobolev embeddings. Thus for all $k\in \NN_0$ there is $s\in \NN_0$ and a constant $C$ such that
 $\|Y(\lambda)f\|_{K,k}\le C \|f\|_{K,s}$ for all $f\in C^\infty_K$. This observation allows the following conclusion to be drawn.
 \begin{Corollary}\label{cor:CKinfty_inversion}
 	In the notation of Theorem~$\ref{thm:Sobolevinversion}$, and for $\lambda\in\CC\setminus S$, $I+ Y(\lambda)$ restricts to a~homeomorphism of $C^\infty_K$, with inverse to be denoted $R(\lambda)$. The map $\lambda\mapsto R(\lambda)$ is
 	holomorphic on $\CC\setminus S$ in the topology of $\LL_b\big(C^\infty_K\big)$.
 	If, additionally, for some holomorphic function $f\colon \CC\setminus S\to C^\infty_K$
 	there is a compact subset $K_f\subset K$ so that $\supp Y(\lambda)^r f(\lambda)\subset K_f$ for all $r\in\NN_0$ and $\lambda\in \CC\setminus S$, then $\supp R(\lambda) f(\lambda)\subset K_f$ for all $\lambda\in \CC\setminus S$. If $Y(\lambda)C^\infty_K\subset C^\infty_{K'}$ for some compact $K'\subset K$, then
 	$\ker(I+ Y(\lambda))|_{C^\infty_K}$ is a finite-dimensional subspace of $C^\infty_{K'}$
 	for all $\lambda\in\CC$.
 \end{Corollary}
 \begin{proof}
 	Suppose $f\in C^\infty_K$ and $\lambda\notin S$. Then, for all $s\ge s_*$,
 	we have	$f\in H^s_K$ and $\big(I+ \hat{Y}_{s}(\lambda)\big)^{-1}f=\big(I+\hat{Y}_{s_*}(\lambda)\big)^{-1}f$, and consequently
 	\begin{equation*}
 		\big(I+\hat{Y}_{s_*}(\lambda)\big)^{-1}f\in \bigcap_{s\ge s_*} H^s_K = C^\infty_K.
 	\end{equation*}
 	Thus $\big(I+\hat{Y}_{s_*}(\lambda)\big)^{-1}$ restricts to linear self-map of $C^\infty_K$ which
 	\big(because $Y(\lambda)$ and $\hat{Y}_{s_*}(\lambda)$ agree on $C^\infty_K$\big) is a linear inverse to the restriction of $I+ Y(\lambda)$. Accordingly, $I+ Y(\lambda)$ restricts to a~continuous bijection of $C^\infty_K$, and since the latter is a Fr\'echet space the inverse mapping theorem implies that the inverse $R(\lambda)$ is continuous.
 	As every $C^\infty_K$ semi-norm is dominated by a~Sobolev norm and
 	vice versa, convergence in the operator norm of every $H^s_K$, for $s\ge s_*$,
 	implies convergence in $\LL_b\big(C^\infty_K\big)$, by Corollary~\ref{cor:HstoCinfty}. It follows that $\lambda\mapsto R(\lambda)$ is holomorphic on~$\CC\setminus S$ in the topology of $\LL_b\big(C^\infty_K\big)$.
 	The remaining statements are immediate from parts~$(a)$ and $(c)$ of Theorem~\ref{thm:Sobolevinversion} and the fact that distributional support of a smooth function is exactly its usual support.
 \end{proof}

 We need the following elementary observation, which will be used for $F=C_K^\infty$, $G=C_0^\infty$.
 \begin{Lemma}\label{lem:inversetransfer}
 	Let $F$ and $G$ be topological vector spaces and suppose that the diagram
 	\begin{equation*} 
 		\begin{tikzcd}
 			F \arrow[r, "S"]\arrow[d,"\iota"] & F \arrow[d,"\iota"]\\
 			G \arrow[r, "T"]\arrow[ur,"\hat{T}"] & G
 		\end{tikzcd}
 	\end{equation*}
 	of continuous linear maps commutes. If $\id_F+S$ is continuously invertible, then
 	$\id_G+T$ is continuously invertible with inverse
 	\begin{equation*}
 		(\id_G+ T)^{-1} = \id_G-T+\iota S(\id_F+S)^{-1}\hat{T}.
 	\end{equation*}
 \end{Lemma}
 \begin{proof}
 	Assuming that $\id_F+S$ is continuously invertible, we compute
 	\begin{align*}
 		(\id_G+T)\big(\id_G-T+\iota S(\id_F+S)^{-1}\hat{T}\big) = \id_G-T^2 + \iota (\id_F+S)S(\id_F+S)^{-1}\hat{T}=\id_G
 	\end{align*}
 	and
 	\begin{align*}
 		\big(\id_G-T+\iota S(\id_F+S)^{-1}\hat{T}\big)(\id_G+T) &= \id_G-T^2 + \iota S(\id_F+S)^{-1}(\id_F+S)\hat{T}= \id_G.
 	\end{align*}
 	using in both cases the identities $\iota S=T\iota$ and hence $\iota S \hat{T} = T\iota\hat{T}=T^2$, and also $\hat{T}T = \hat{T}\iota\hat{T} = S\hat{T}$. Therefore $\id_G+T$ has a linear inverse, given by a manifestly continuous expression.
 \end{proof}

 We come to the proof of the main result.
 \begin{proof}[Proof of Theorem~\ref{thm:GHpert}] The proof involves several steps, and uses some technical lemmas from the Appendix.
 	
 	\emph{$1.$ Preliminary observations.} Define $Y^\pm(\lambda)\in\LL(C_0^\infty,C^\infty)$ by
 	\begin{equation*}
 		Y^\pm(\lambda)=A(\lambda)E^\pm.
 	\end{equation*}
 	Owing to assumptions $(a)$ and $(b)$, we may extend $Y^\pm(\lambda)$ to linear maps $\EE'\to\DD'_K$ with continuous restrictions mapping $H^s_0\to H^{s+\delta}_K$ for all $s\ge s_*$; hence there are also continuous restrictions $Y_s^\pm(\lambda)\colon H^s_K\to H^{s+\delta}_K$ for $s\ge s_*$,
 	and by the Sobolev embedding theorems, compact maps $\hat{Y}^\pm_s(\lambda)\colon H^s_K\to H^{s}_K$ for $s\ge s_*$ which are holomorphic with respect to the operator norm topology, as noted after the statement of Theorem~\ref{thm:GHpert}.
 	
 	In fact, if $K'$ is any compact topologically regular set, the same argument shows that~$Y^\pm(\lambda)$
 define compact maps depending holomorphically on $\lambda$ in the topology of
 	$\LL_b\big(H^s_{K\cup K'}\big)$ for all~${s\!\ge\!s_*}$. By Corollary~\ref{cor:HstoCinfty}, it follows that
 	$Y^\pm(\lambda)$ are holomorphic in $\lambda$ with respect to the $\LL_b\big(C^\infty_{K\cup K'}\big)$ topology.
 	As $C^\infty_K$ and $C^\infty_{K'}$ are continuously embedded topological subspaces of $C^\infty_{K\cup K'}$, and $\Ran Y^\pm(\lambda)\subset C^\infty_K$, we may use Lemma~\ref{lem:Lb_basic}\,$(a)$,\,$(c)$ to deduce that
 	the $Y^\pm(\lambda)$ are holomorphic in $\lambda$ with respect to the
 	topology of $\LL_b\big(C^\infty_{K'},C^\infty_{K}\big)$.
 	By Lemma~\ref{lem:LbLF}, it follows that
 	$A(\lambda)E^\pm$ are also holomorphic with respect to $\LL_b\big(C_0^\infty,C^\infty_K\big)$ and $\LL_b(C_0^\infty)$; using Lemma~\ref{lem:Lb_basic}\,$(a)$, we also have holomorphicity with respect to
 	$\LL_b\big(C^\infty_K,C_0^\infty\big)$.
 	
 	\emph{$2.$ Finite-dimensionality of $\ker (P+ A(\lambda))|_{C^\infty_{\pc/\fc}}$.} Next, observe that $P$ induces bijections between $\ker (P+ A(\lambda))|_{C^\infty_{\pc/\fc}}$ and $\ker \big(I+ Y^{\pm}(\lambda)\big)\big|_{C^\infty_K}$,
 	with inverses given by the restrictions of~$E^{\pm}$. For
 	if $(P+ A(\lambda))\phi=0$ with $\phi\in C^\infty_{\pc/\fc}$ then
 	$P\phi=- A(\lambda)\phi\in C^\infty_K$ and $\phi=E^\pm P\phi$, so $P\phi\in\ker\big(I+ A(\lambda)E^\pm\big)\big|_{C^\infty_K}$; conversely, if $\big(I+ A(\lambda)E^\pm\big)h=0$ with
 	$h\in C^\infty_K$, then $PE^\pm h = h=- A(\lambda)E^\pm h$, so
 	$E^\pm h\in \ker(P+ A(\lambda))|_{C^\infty_{\pc/\fc}}$.
 	Thus
 	\begin{equation*}
 		S^\pm:= \big\{\lambda\in \CC\colon \ker (P+ A(\lambda))|_{C^\infty_{\pc/\fc}}\neq 0\big\}= \big\{\lambda\in \CC\colon \ker \big(I+ Y^{\pm}(\lambda)\big)\big|_{C^\infty_K}\neq 0\big\}.
 	\end{equation*}
 	In combination with Theorem~\ref{thm:Sobolevinversion}\,$(a)$, we also have
 	\begin{equation}\label{eq:GHpert_remark}
 		\dim \ker (P+A(\lambda))|_{C^\infty_{\pc/\fc}} = \dim \ker \big(I+Y_s^{\pm}(\lambda)\big)\big|_{C^\infty_{K}} = \dim \ker \big(I+\hat{Y}_s^{\pm}(\lambda)\big)<\infty
 	\end{equation}
 	\big(the latter kernel taken in $H^s_K$\big) for all $s\ge s_*$. Part~$(A)$ of the Theorem is thus proved.
 	
 	\emph{$3.$ Construction of holomorphic candidate $K$-nonlocal Green operators.} Applying Theorem~\ref{thm:Sobolevinversion} and Corollary~\ref{cor:CKinfty_inversion} to $Y^\pm(\lambda)$, one finds that $S^\pm$ are discrete subsets of $\CC$, whose complements in~$\CC$ are open $0$-neighbourhoods, and the operators $I+ Y^\pm(\lambda)$ are continuously invertible on~$C^\infty_K$ for $\lambda\in\CC\setminus S^\pm$, with inverses
 	that are holomorphic on~$\CC\setminus S^\pm$ in the topology
 	of $\LL_b\big(C^\infty_K\big)$. All these properties hold also for $S=S^+\cup S^-$,
 	so part~(B) is proved.
 	
 	Fixing any $\lambda\in \CC\setminus S$, the diagram of continuous maps
 	\begin{equation*}
 		\begin{tikzcd}
 			C^\infty_K \arrow[r, "T_K"]\arrow[d] & C^\infty_K \arrow[d]\\
 			C_0^\infty \arrow[r, "T_0"]\arrow[ur,"\hat{T}_0"] & C_0^\infty
 		\end{tikzcd}
 	\end{equation*}
 	commutes, where the unlabelled arrows are the canonical inclusions,
 	$T_0$ and $T_K$ are the respective restrictions of $Y^\pm(\lambda)$
 	to $C_0^\infty$, and $C_K^\infty$ and $\hat{T}_0$ exists because
 	$A(\lambda)C^\infty\subset C^\infty_K$.
 	It follows from Lemma~\ref{lem:inversetransfer} that $I+ Y^\pm(\lambda)$ are continuously invertible on $C_0^\infty$. By abuse of notation we write the inverses as $\big(I+ A(\lambda)E^\pm\big)^{-1}$; Lemma~\ref{lem:inversetransfer} now gives the identity
 	\begin{equation}\label{eq:inverseidentity}
 		\big(I+ A(\lambda)E^\pm\big)^{-1} = I - A(\lambda)E^\pm + A(\lambda)E^\pm \big(I+ A(\lambda)E^\pm\big)^{-1} A(\lambda)E^\pm,
 	\end{equation}
 	on $C_0^\infty(K)$, where we suppress notation for inclusions and restrictions. Note that the inverse on the right-hand side is taken in $\LL\big(C^\infty_K\big)$, while the left-hand side is an inverse in $\LL(C^\infty_0)$. Because the former inverse is holomorphic in
 	$\LL_b\big(C^\infty_K\big)$, the Leibniz rule (see Corollary~\ref{cor:Leibniz}) implies that the left-hand side is
 	holomorphic in $\LL_b(C_0^\infty)$; here, we have also used the holomorphicity
 	of~$A(\lambda)E^\pm$ in $\LL_b\big(C_0^\infty,C^\infty_K\big)$, $\LL_b(C_0^\infty)$ and $\LL_b\big(C^\infty_K,C_0^\infty\big)$ established in step 1 of the proof. It follows that the operators
 	\begin{equation*}
 		\tilde{E}^\pm_\lambda= E^\pm \big(I+ A(\lambda)E^\pm\big)^{-1}\in\LL(C_0^\infty, C^\infty),
 	\end{equation*}
 	which are the candidate Green operators for $P+A(\lambda)$,	are holomorphic in $\lambda$ on $\CC\setminus S$ with respect to the topology of $\LL_b(C_0^\infty,C^\infty)$.
 	To prove part (C) it is now enough to check that $\tilde{E}^\pm_\lambda$ are indeed $K$-nonlocal Green operators.
 	
 	\emph{$4.$ Verification that $\tilde{E}^\pm_\lambda$ are $K$-nonlocal Green operators.}
 	Given any $f\in C_0^\infty$,
 	\begin{equation*}
 		g=\big(I+ A(\lambda)E^\pm\big)^{-1} f
 	\end{equation*}
 	is the unique element of $C_0^\infty$ obeying
 	\begin{equation}\label{eq:geqn}
 		g + A(\lambda) E^\pm g= f,
 	\end{equation}
 	whereupon we deduce that $\supp g\subset \supp f \cup K$ and that
 	\begin{equation*}
 		\varphi = E^\pm g = E^\pm \big(I+ A(\lambda)E^\pm\big)^{-1} f = \tilde{E}^\pm_\lambda f
 	\end{equation*}
 	satisfies
 	\begin{equation}\label{eq:Pphi}
 		P\varphi + A(\lambda)\varphi = g+A(\lambda)E^\pm g=f, \qquad \supp \varphi\subset J^\pm(\supp g) \subset J^\pm(K\cup \supp f).
 	\end{equation}
 	In the special case $f=(P+A(\lambda))h$, for $h\in C_0^\infty$,
 	we have
 	\begin{equation*}
 		f = (P+ A(\lambda))E^\pm P h=Ph + A(\lambda) E^\pm Ph
 	\end{equation*}
 	and the unique solution to \eqref{eq:geqn} is clearly $g=Ph$. Thus
 	\begin{equation*}
 		Ph=g=\big(I+ A(\lambda)E^\pm\big)^{-1}f =\big(I+ A(\lambda)E^\pm\big)^{-1}(P+ A(\lambda))h;
 	\end{equation*}
 	consequently $E^\pm \big(I+ A(\lambda)E^\pm\big)^{-1} (P+ A(\lambda))h = h$.
 	In combination with~\eqref{eq:Pphi}, we have shown
 	\begin{equation*}
 		(P+ A(\lambda))\tilde{E}^\pm_\lambda f=f = \tilde{E}^\pm_\lambda (P+ A(\lambda)) f,\qquad\text{and}\qquad
 		\supp\tilde{E}^\pm_\lambda f\subset J^\pm(K\cup \supp f)
 	\end{equation*}
 	for all $f\in C_0^\infty$.
 	Now suppose more specifically that $J^\pm(\supp f)\cap K=\varnothing$ for $f\in C^\infty_0$. Then $A(\lambda)E^\pm f=0$, due to assumption $(d)$, and~\eqref{eq:geqn} is solved by $g=f$, so
 	$\tilde{E}^\pm_\lambda f=E^\pm f$ has support contained in $J^\pm(\supp f)$. Accordingly,
 	\begin{equation*}
 		\supp\tilde{E}^\pm_\lambda f\subset
 		\begin{cases}
 			J^\pm(\supp f), & J^\pm(\supp f)\cap K=\varnothing ,\\
 			J^\pm(K\cup \supp f), & \text{otherwise,}
 		\end{cases}
 	\end{equation*}
 	so $\tilde{E}^\pm_\lambda$ are $K$-nonlocal Green operators for $P+ A(\lambda)$. Part $(C)$ is complete.
 	
 	\emph{$5.$ Continuous extension.} Due to~\eqref{eq:inverseidentity}, one has the formula
 	\begin{equation*}
 		\tilde{E}_\lambda^\pm = E^\pm - E^\pm A(\lambda)E^\pm + E^\pm A(\lambda)E^\pm \big(I+ A(\lambda)E^\pm\big)^{-1}A(\lambda)E^\pm,
 	\end{equation*}
 	in which the three terms on the right-hand side have continuous
 	extensions as maps from $H^s_0$ to $H^{s+\beta}_\loc$, $H^{s+\beta+\delta}_\loc$ and $H^{s+\beta+2\delta}_\loc$ respectively. Because $\delta>0$, we deduce that $\tilde{E}^\pm_\lambda$
 	extends continuously to a map $H^s_0\to H_\loc^{s+\beta}$ as required.
 	By the Leibniz rule (Corollary~\ref{cor:Leibniz}), this extension is holomorphic in $\lambda$ with respect to the topology of $\LL_b\big(H^s_0,H_\loc^{s+\beta}\big)$.
 	This proves part (D).

 	\emph{$6.$ Support non-increasing modifications.} Finally, suppose condition $(d')$ holds, so that one has $\supp A(\lambda)f\subset \supp f$ as well as
 	$\supp A(\lambda)f\subset K$ for all $f\in C^\infty$. Then $(d)$ also holds,
 	either as a consequence of Peetre's theorem~\cite{Peetre:1960} or by the
 	following direct argument:
 	if $f\in C^\infty$ vanishes identically on $K$, then the carrier of $A(\lambda)f\in C^\infty_K$ satisfies
 	both
 	\begin{equation*}
 		\carr A(\lambda)f \subset \inte(K)
 	\end{equation*}
 	and
 	\begin{equation*}
 		\carr A(\lambda)f\subset \supp A(\lambda)f\subset \supp f=\overline{\carr f}\subset \overline{M\setminus K} = M\setminus \inte(K)
 	\end{equation*}
 	and we deduce that $\carr A(\lambda)f$ is empty, i.e., $A(\lambda)f$ vanishes identically. Therefore, all the conclusions reached previously still hold.
 	
 	It further follows that
 	\begin{equation*}
 		\supp Y^\pm(\lambda) f =\supp A(\lambda)E^\pm f\subset K\cap \supp E^\pm f\subset K\cap J^\pm(\supp f)
 		\qquad \text{for all} \ f\in C_0^\infty;
 	\end{equation*}
iterating, we have in particular that
 	\begin{equation*}
 		\supp \big(Y^\pm(\lambda)\big)^r f\subset K^\pm_f:= K\cap J^\pm(\supp f) \qquad \text{for all} \ r\in \NN,\ f\in C_0^\infty \ \text{and} \ \lambda\in\CC\setminus S.
 	\end{equation*}	
 	Therefore $\big(Y^\pm(\lambda)\big)^r Y^\pm(\lambda) f\in C^\infty_{K^\pm_f}\subset C^\infty_K$ for all $r\in\NN_0$, $f\in C_0^\infty$, $\lambda\in \CC\setminus S$.
 	Using Corollary~\ref{cor:CKinfty_inversion}, applied to the function $\lambda\mapsto Y^\pm(\lambda) f\in C^\infty_K$, it follows that both
 	$\supp \big(I+Y^\pm(\lambda)\big)^{-1} Y^\pm(\lambda) f$ and
 	\begin{equation*}
 		\supp Y^\pm(\lambda) \big(I+Y^\pm(\lambda)\big)^{-1} Y^\pm(\lambda) f= \supp \big(I-(I+ Y^\pm(\lambda))^{-1}\big) Y^\pm(\lambda) f
 	\end{equation*}
 	are contained in $K^\pm_f$ for all $\lambda\in \CC\setminus S$. Using the identity~\eqref{eq:inverseidentity},
 	it now follows that
 	\begin{equation*}
 		\supp \big(I+ A(\lambda)E^\pm\big)^{-1} f \subset \supp f\cup \left(K\cap J^\pm(\supp f)\right) \subset
 		J^\pm(\supp f)
 	\end{equation*}
 	and hence $\supp \tilde{E}^\pm_\lambda f \subset J^\pm\big(J^\pm(\supp f)\big)=J^\pm(\supp f)$, which is the support property (G3) for standard Green operators, completing the proof of part $(E)$ and therefore the whole theorem.
 \end{proof}

 \section{Proof of Theorem~\ref{thm:Fredholm}}\label{sec:Fredholm}

 The aim of this section is to prove a relation between the spontaneously appearing and disappearing solutions for the operators $P+A(\lambda)$ and $\adj{P}+\adj{A}(\lambda)$ stated in Theorem~\ref{thm:Fredholm}.

 Starting with some preliminaries, let $E^\pm$ and $E_t^\pm$ be the
 Green operators for $P$ and $\adj{P}$.
 Because $s_*\le-\beta\le 0$, $E^\pm$ and $E_t^\pm$ have extensions from $C^\infty_K$ to continuous maps in both $\LL\big(H^0_0,H^\beta_\loc\big)$ and $\LL\big(H^\gamma_0,H^\delta_\loc\big)$, while $A(\lambda)$ and $\adj{A}(\lambda)$ have extensions from $C^\infty$ to continuous maps in both $\LL\big(L^2_\loc,H^\gamma_K\big)$ \big(and consequently $\LL\big(L^2_K,H^\gamma_K\big)$\big) and $\LL\big(H^\beta_\loc,H^\delta_K\big)$.
 We also write $Y^\pm(\lambda)=A(\lambda)E^\pm$ and $Y_t^\pm = \adj{A}(\lambda) E_t^\pm$
 as compact operators on any $H^s_K$ \big(the value of $s$ will be clear from context, and we suppress embedding maps between various $H^s_K$ spaces\big).

 We generally abuse notation by using the same notation for $A(\lambda)$ whether it operates on~$H^s_K$ or~$H^s_\loc$ and regardless of $s$, but our arguments will take proper account of the domains concerned.
 We make use of two technical facts in the $s=0$ case.
 \begin{Lemma}\label{lem:tech}
 	Under the stated assumptions on $A(\lambda)$ and $Y^\pm(\lambda)$:
 \begin{enumerate}\itemsep=0pt
 \item[$(a)$] the identity
 	\begin{equation*}
 		A(\lambda)f = A(\lambda)(f|_K)
 	\end{equation*}
 	holds for all $f\in L^2_\loc$, where we understand the map $A(\lambda)\in \LL\big(L^2_\loc,H^\gamma_K\big)$ on the left-hand side and $A(\lambda)\in \LL\big(L^2_K,H^\gamma_K\big)$ on the right-hand side;
 \item[$(b)$]
 	for the operators $Y_t^\pm(\lambda)\in \LL\big(L^2_K,L^2_K\big)$, we have
 	\begin{equation*}
 		\big(Y_t^\pm(\lambda)\big)^*f = (\Gamma E^\mp A(\lambda)\Gamma f )|_K, \qquad f\in L^2_K,
 	\end{equation*}
 	where $\Gamma$ denotes complex conjugation.
 \end{enumerate}
 \end{Lemma}	
 \begin{proof}
 	We suppress the $\lambda$ dependence in this proof. For part $(a)$,
 	first choose $f_n\in C^\infty$ with $f_n\to f$ in $L^2_\loc$. Then for any
 	$g\in C_0^\infty$ the distributional action of $Af$ on $\mu g$ (recalling that $\mu$ is the volume density) is given by
 	\begin{equation*}
 		(Af)(\mu g)=\lim_{n\to\infty} \int_M \mu g (Af_n) =
 		\lim_{n\to\infty}\int_M \mu (\adj{A}g) f_n = \ip{\overline{\adj{A}g}}{f|_K},
 	\end{equation*}
 	where the inner product $\langle\cdot\mid\cdot\rangle$ is that of $L^2_K$ and we use the fact that $\adj{A}g\in C^\infty_K$. Now choose a~new sequence $f_n\in C^\infty_K$ with $f_n\to f|_K$ in $L^2_K$. Then
 	\begin{equation*}
 		(Af)(\mu g)= \ip{\overline{\adj{A}g}}{f|_K}=\lim_{n\to\infty}
 		\int_M \mu \big(\adj{A}g\big)f_n = \lim_{n\to\infty} \int_M \mu g(Af_n) = (Af|_K)(\mu g).
 	\end{equation*}
 	As $g$ was arbitrary, $Af$ and $Af|_K$ define the same distribution and hence the same element of~$H^\gamma_K$.
 	
 	For part $(b)$, we compute for $f,h\in C^\infty_K$ that
 	\begin{align*}
 		\ip{\big(Y_t^\pm\big)^*f}{h} &= \ip{f}{Y_t^\pm h}
 		= \int_M\mu \overline{f} (
 		\adj{A}) E_t^\pm h = \int_M\mu \big(E^\mp A \overline{f} \big) h = \ip{(\Gamma E^\mp A \Gamma f)|_K}{h}.
 	\end{align*}
 	Thus we have $\big(Y_t^\pm(\lambda)\big)^*f = (\Gamma E^\mp A \Gamma f)|_K$ for
 	all $f\in C^\infty_K$ and hence for all $f\in L^2_K$ by continuity.\looseness=1
 \end{proof}

 \begin{proof}[Proof of Theorem~\ref{thm:Fredholm}]
 	Let $\adj{N}^\pm(\lambda) =\dim\ker (P+A(\lambda))|_{C^\infty_{\pc/\fc}}$.
 	By the remark~\eqref{eq:GHpert_remark} in the proof of Theorem~\ref{thm:GHpert}, applied to $Y_t^\pm$ in the case $s=0$, we have $\adj{N}^\pm(\lambda) = \dim \ker \big(I+Y_t^\pm(\lambda)\big)$ in $L^2_K$. As $Y_t^\pm$ are compact and holomorphic in the suppressed parameter $\lambda$,
 	Fredholm theory implies that the index
 	\begin{equation*}
 		\Ind \big(Y_t^\pm(\lambda)\big)=
 		\dim \ker\big(I+Y_t^\pm(\lambda)\big) - \dim\ker\big(I+\big(Y_t^\pm(\lambda)\big)^*\big)
 	\end{equation*}
 	is independent of $\lambda$ and therefore vanishes on considering $\lambda=0$~-- see, e.g., \cite[Theorem~4.3.12]{Davies_LOps}~-- so
 	\begin{equation*}
 		\adj{N}^\pm(\lambda) = \dim\ker\big(I+\big(Y_t^\pm(\lambda)\big)^*\big) .
 	\end{equation*}
 	Below, we will show that $A\Gamma$ induces an antilinear injection between the kernels of
 	$I+\big(Y_t^\pm(\lambda)\big)^*$ and $I+Y^\mp(\lambda)$ in $L^2_K$; swapping the roles of $P$, $A(\lambda)$ and $\adj{P}$, $\adj{A}(\lambda)$,
 	we find that the spaces have equal (finite) dimension. Consequently,
 	\begin{equation*}
 		\adj{N}^\pm(\lambda) = \dim\ker (I+Y^\mp(\lambda)) = N^\mp(\lambda),
 	\end{equation*}
 	again using remark~\eqref{eq:GHpert_remark} in the proof of Theorem~\ref{thm:GHpert}, which is the desired result.
 	
 	It remains to prove that $A\Gamma$ provides the required antilinear injection.
 	Suppose that
 	there is $f\in L^2_K\setminus\{0\}$ with $\Gamma f=-\big(Y_t^\pm(\lambda)\big)^*\Gamma f$ for some fixed $\lambda$. Then by Lemma~\ref{lem:tech}\,$(b)$
 	\begin{equation*}
 		f = -(E^\mp A(\lambda) f) |_K
 	\end{equation*}
 	and since $f\neq 0$ it also follows that $A(\lambda)f\neq 0$ (otherwise $f=-E^\mp A(\lambda)f|_K=0$). Furthermore,
 	\begin{equation*}
 		A(\lambda) f =- A(\lambda) (E^\mp A(\lambda) f)|_K =- A(\lambda) E^\mp A(\lambda) f
 		=-Y^\mp(\lambda)A(\lambda)f
 	\end{equation*}
 	holds in $H^\gamma_K$, so $A(\lambda)f \in \ker (I+Y^\mp(\lambda))$ in $H^\gamma_K$,
 	and therefore also in $L^2_K$ (since $\gamma>0\ge s_*$ the kernels are equal by Theorem~\ref{thm:Sobolevinversion}\,$(a)$). Accordingly, $A(\lambda)\Gamma$ is an antilinear injection from the $L^2_K$ kernel of
 	$I+\big(\adj{Y}{}^\pm(\lambda)\big)^*$ to the $L^2_K$ kernel of $I+Y^\mp(\lambda)$.
 \end{proof}

 \appendix
 \section{Some topological vector spaces}\label{appx:topspaces}

 We briefly rehearse the definition and main properties of the various $C^k$ and Sobolev spaces encountered in the text, broadly following~\cite{Baer:2015,BaerWafo:2015}, before turning to some properties of the topology of bounded convergence that are also needed. No originality is claimed for the material given here.

 \subsection[Spaces of smooth and C\^{}k functions]{Spaces of smooth and $\boldsymbol{C^k}$ functions}
 Let $M$ be a smooth manifold and let $C^k(M)$, $k\in\NN_0\cup\{\infty\}$, be the vector space of complex-valued $k$-times continuously differentiable functions on $M$. For each $k\in\NN_0$ and compact $K\subset M$, one has a seminorm
 \begin{equation*}
 	\|f\|_{K,k} = \max_{0\le r\le k} \max_{x\in K} |(\nabla^r f)(x)|_r
 \end{equation*}
 on $C^k(M)$, where $\nabla$ is an arbitrarily chosen connection on $M$ and $|\cdot|_r$ an arbitrarily chosen norm making $T^*M^{\otimes r}$ a (finite-dimensional) Banach bundle; different choices result in equivalent seminorms. The collection of seminorms $\|\cdot\|_{K,k}$ as $K$ runs over compact subsets of $M$ and $k\in\NN_0$ provides a Fr\'echet topology on $C^\infty(M)$; similarly,
 we obtain a Fr\'echet topology on~$C^k(M)$, $k\in\NN_0$, using the seminorms
 $\|\cdot\|_{K,k}$.

 If $A$ is closed, we define $C^\infty_A(M)=\{f\in C^\infty(M)\colon\supp f\subset A\}$ with the relative topology. Thus the topology is defined by the seminorms $\|\cdot\|_{K,k}$ as $K$ runs over compact subsets $K\subset A$ and $k\in\NN_0$; if $A$ is compact, it is sufficient to use the seminorms $\|\cdot\|_{A,k}$, $k\in\NN_0$. As $C^\infty_A(M)$ is closed subspace of a Fr\'echet space, it is also Fr\'echet. Defining $C^k_A(M)$ as the analogous subspace of $C^k(M)$, the topology is generated by seminorms $\|\cdot\|_{K,k}$ for compact $K\subset A$, or just the single seminorm $\|\cdot\|_{A,k}$ (which is a norm on $C^k_A(M)$) in the case that $A$ is compact.

 A \emph{support system}~\cite{Baer:2015} is a subset $\mathscr A$ of the set of all closed subsets on $M$, which is closed under finite unions and has the property that for each $A\in \mathscr A$ it holds that $(i)$
 $A\subset \textrm{int}\,(A')$ for some $A'\in\mathscr A$ and $(ii)$ if $A'$ is a closed subset of $M$ with $A'\subset A$ then $A'\in\mathscr A$. Any support system is a directed system with respect to inclusion and we write
 \begin{equation*}
 	C^\infty_{\mathscr{A}}(M)= \bigcup_{A\in\mathscr{A}} C^\infty_A(M)
 \end{equation*}
 with the locally convex inductive limit topology,
 so that a convex set $U\subset C^\infty_{\mathscr{A}}(M)$ is a neighbourhood of $0$
 if and only if $U\cap C^\infty_A(M)$ is a neighbourhood of $0$ in $C^\infty_A(M)$ for every $A\in\mathscr{A}$; because $\mathscr{A}$ is directed, one also has that
 a convex set $O$ is open if and only if $O\cap A$ is open in each $C^\infty_A(M)$.

 Examples of support systems include the set of compact sets, leading to the space of compactly supported functions $\CoinX{M}$, and the sets of (strictly) future/past/spatially-compact sets, giving rise to $C^\infty_{\sfc/\spc/\fc/\pc/\mathrm{sc}}(M)$.

 A linear map $T$ from $C^\infty_{\mathscr{A}}(M)$ to a locally convex topological vector space $Y$ is continuous if and only if all its restrictions $T_A\colon C^\infty_A(M)\to Y$ are continuous ($A\in\mathscr{A}$). In particular, for a linear map $T\colon C^\infty_{\mathscr{A}}(M)\to C^\infty_{\mathscr{B}}(M)$ to be continuous, it suffices that each restriction $T_A$ has range contained in some $C^\infty_B(M)$ for some $B\in\mathscr{B}$ (depending on $A$) and determines a continuous map $C^\infty_A(M)\to C^\infty_B(M)$, thus implying that $T_A\colon C^\infty_A(M)\to C^\infty_{\mathscr{B}}(M)$ is continuous. We record the following application:
 \begin{Lemma}
 	If $T$ is a continuous linear self-map of $C^\infty(M)$ such that $\supp Tf\subset K\cup \supp f$ for all $f\in C^\infty(M)$, where $K$ is a fixed compact set, then $T$ restricts to any of $C^\infty_{0/\sfc/\spc/\fc/\pc/\mathrm{sc}}(M)$ as a continuous map.
 \end{Lemma}
 \begin{proof}
 	If $\mathscr{A}$ is one of the relevant support systems then $A\cup K\in\mathscr{A}$ for each $A\in\mathscr{A}$. As the restriction $T_A$ of $T$ to $C^\infty_A(M)$ has its range in $C^\infty_{A\cup K}(M)$, and both these function spaces are subspaces of $C^\infty(M)$ with the relative topology, continuity of $T_A\colon C^\infty_A(M)\to C^\infty_{A\cup K}(M)$ follows from that of $T$. Letting $A$ vary in $\mathscr{A}$, the result is proved.
 \end{proof}

 Letting $\Omega_\alpha$ be the bundle of densities of weight $\alpha$, we can define spaces $\Gamma^\infty(\Omega_\alpha)$,
 $\Gamma_A^\infty(\Omega_\alpha)$ and $\GoinX{\Omega_\alpha}$ of smooth sections
 of $\Omega_\alpha$ in a similar way; any smooth nowhere vanishing density $\mu$ on
 $M$ provides topological isomorphisms between these spaces and their zero-weight analogues, simply by multiplication by appropriate powers of $\mu$. The space of distributions on $M$, $\DD'(M)$ is the topological dual of $\GoinX{\Omega_1}$,
 equipped with the weak-$*$ topology. In particular, $C^\infty(M)$ is canonically embedded as a subspace of $\DD'(M)$. More generally,
 $\DD'(M,\Omega_\alpha)$ is defined as the dual of $\GoinX{\Omega_{1-\alpha}}$.
 We also write $\EE'(M)$ for the topological dual of $\Gamma^\infty(\Omega_1)$, i.e.,
 the distributions of compact support.

 \subsection{Sobolev spaces}

 \subsubsection{Compact manifolds}

 {\bf The spaces $\mathbf{H^s(M)}$.}
 On a \emph{compact} manifold $M$, choose an auxiliary Riemannian metric and define $L^2(M)$ in the usual way, using the volume element induced by the metric and writing the inner product as $\langle\cdot\mid\cdot\rangle$ (linear in the second slot). Let $T=(-\triangle + I)^{1/2}$, where $\triangle$ is the Laplace--Beltrami operator, initially defined on $C^\infty(M)$ and then extended (uniquely) to a self-adjoint negative operator on $L^2(M)$.
 Its complex powers $T^s$ (as defined by functional calculus) are classical pseudodifferential operators of order $s$, $T^s\in\Psi^s(M)$~\cite{Seeley:1967} (see~\cite{Amman_etal:2004} for an axiomatically based proof) and $T^s$ is compact for $\Re s<0$ (most easily seen using the spectral properties of~$-\triangle$ on compact manifolds). The domain of $T^s$ contains $C^\infty(M)$ for all $s$.

 For $s\in\RR$, the Sobolev space $H^s(M)$ is defined as the completion of $C^\infty(M)$ with respect to the norm
 \begin{equation*}
 	\|f\|_{H^s(M)} = \| T^s f\|_{L^2(M)}
 \end{equation*}
 and, as a topological vector space, is independent of the choice of auxiliary Riemannian metric involved in its construction (of course, the specific norm $\|\cdot\|_{H^s(M)}$
 and its compatible Hilbert space inner product
 are metric-dependent). Indeed, any positive second-order elliptic operator that is essentially self-adjoint on $C^\infty(M)$ could be used in place of $-\triangle$.
 Evidently $T^s$ extends to an isometry from $H^s(M)$ to $L^2(M)$, which may be used to embed $H^s(M)$ in $\DD'(M)$ so that $u \in H^s(M)$ corresponds to the distribution $f\mapsto \ip{T^{-s}\overline{f/\nu}}{T^su}$, $f\in\GoinX{\Omega_1}$, where $\nu$ is the density corresponding to the
 volume element of the auxiliary metric used to define $L^2(M)$. Restricted to $u\in C^\infty(M)$, this embedding is consistent with the embedding of $C^\infty(M)$ in $\DD'(M)$ already mentioned.
 As $T^q$ is compact on $L^2(M)$ for $q<0$, it follows that $H^s(M)\subset H^t(M)$ for all $s>t$, with a compact inclusion map. More generally. $P\colon H^{s+m}(M)\to H^s(M)$ is continuous for any partial differential operator of order $m$ with smooth coefficients (or indeed any pseudodifferential operator $P\in \Psi^{m}(M)$), because $T^{-m}P\in \Psi^0(M)$ extends to a bounded operator on $L^2(M)$.

 {\bf Duality.} For any $s\in\RR$, suppose that $\ell\in H^s(M)'$, so there is a constant $c$ so that $|\ell(u)|\le c\|u\|_{H^s(M)}$ for all $u\in H^s(M)$. Then $|\ell(T^{-s}f)|\le c\|f\|_{L^2(M)}$ for $f\in L^2(M)$ and consequently there is $w\in L^2(M)$ so that $\ell(u) =\langle w\,|\, T^{s}u\rangle_{L^2(M)}$ for all $u\in H^s(M)$. Choose a sequence $w_n\to w$ in $L^2(M)$ with
 $w_n\in C^\infty(M)$ and note that $T^s w_n\in C^\infty(M)$ is a Cauchy sequence with respect to the
 $H^{-s}(M)$ norm, converging to some $v\in H^{-s}(M)$ for which
 $T^{-s}v=\lim_n w_n=w$. Thus
 $\ell(u)=\langle T^{-s}v\, |\,T^su\rangle$ for all $u\in H^s(M)$ and
 $|\ell(u)|\le \|v\|_{H^{-s}(M)}\|u\|_{H^s(M)}$. Noting that
 \begin{equation*}
 	\frac{\ell(T^{-s}w_n)}{\|T^{-s}w_n\|_{H^s(M)}}=\frac{\langle w\mid w_n \rangle_{L^2(M)}}{\|w_n\|_{L^2(M)}}\to \|w\|_{L^2(M)} = \|v\|_{H^{-s}(M)},
 \end{equation*}
 we see that $\|\ell\|=\|v\|_{H^{-s}(M)}$, and we have established that
 the space $H^{-s}(M)$ is anti-isomorphic to $H^s(M)'$ with the operator norm topology (i.e.,
 the strong dual) with respect to the sesquilinear pairing $\langle v,u\rangle = \langle T^{-s}v\,|\, T^s u\rangle$, $v\in H^{-s}(M)$, $u\in H^s(M)$. To be precise: every
 choice of Riemannian metric on $M$ induces an anti-isomorphism of this type, and the
 pairing depends on the choice made.

 {\bf Embedding theorems.}
 The relationship between $C^k$ and Sobolev spaces is given as follows. In one direction, the formula $-\triangle=\delta d$ (on $0$-forms) gives the estimates
 \begin{equation*}
 	0\le \ip{f}{(-\triangle)^{2r}f} = \| (-\triangle)^r f\|^2_{L^2(M)}\le C\|f\|_{M,2r}^2
 \end{equation*}
 and
 \begin{equation*}
 	0\le \ip{f}{(-\triangle)^{2r+1}f} = \| d(-\triangle)^r f\|^2_{\Lambda^1(M)}\le C\|f\|_{M,2r+1}^2 ,
 \end{equation*}
 where $\Lambda^1(M)$ is the Hilbert space of square-integrable $1$-forms (with respect to the auxiliary Riemannian metric) and $f\in C^\infty(M)$.
 From this we find
 \begin{equation*}
 	\|f\|_{H^k(M)}^2 = \ip{f}{(-\triangle+I)^k f} = \sum_{j=0}^k \binom{k}{j} \ip{f}{(-\triangle)^{j}f} \le C \|f\|_{M,k}^2
 \end{equation*}
 for all $k\in \NN_0$. Thus $C^k(M)$ is continuously embedded in $H^s(M)$ for all $s\ge k$.

 On the other hand, now let $n$ be the maximum dimension of any component of $M$.
 For $\Re s>n/2$, the integral kernel $Z_s(p,q)$ of the operator $(-\triangle+I)^{-s}=T^{-2s}$ on $L^2(M)$ is continuous, $Z_s\in C(M\times M)$ and can also be written
 in terms of the spectral decomposition of
 $(-\triangle+I)$ as
 \begin{equation*}
 	Z_s(p,q) = \sum_{r} \frac{e_r(p)\overline{e_r(q)}}{\lambda_r^{s}},
 \end{equation*}
 where $e_r\in L^2(M)$ are a basis of smooth orthonormal eigenfunctions for $-\triangle+I$ with corresponding eigenvalues $\lambda_r$ -- see~\cite{Seeley:1967}. In particular,
 \begin{equation*}
 	v_p = \sum_{r} \frac{\overline{e_r(p)}e_r}{\lambda_r^{s/2}}
 \end{equation*}
 belongs to $L^2(M)$ with $\|v_p\|^2 = Z_s(p,p)$ by the Pythagoras theorem.
 For $\Re s>n/2$, it follows that $\|e_r\|_\infty \le C_s\big|\lambda_r^{s/2}\big|$, where $C_s=\sup_{p\in M}|Z_s(p,p)|^{1/2}$. Moreover, if $f\in C^\infty(M)$,
 \begin{equation*}
 	|\langle e_r \mid f \rangle|=\lambda_r^{-k}\big|\ip{e_r}{T^k f}\big|\le \lambda^{-k}_r\|f\|_{H^k(M)},\qquad k\in\NN,
 \end{equation*}
 so $\ip{e_r}{f}$ decays faster than any inverse power of $\lambda_r$ and
 $\ip{v_p}{T^s f} = \sum_r e_r(p)\ip{e_r}{f}= f(p)$ for all (not merely almost all) $p\in M$. Consequently,
 \begin{equation*}
 	\|f\|_\infty =\sup_{p\in M} |\langle v_p \mid f \rangle|\le C_s \|f\|_{H^s(M)}, \qquad f\in C^\infty(M).
 \end{equation*}
 By density of $C^\infty(M)$ in $H^s(M)$, it follows
 that there is a continuous embedding $H^s(M)\to C(M)$ if $s>n/2$; considering $\|Pf\|_\infty$ in a~similar way for differential operators $P$, one sees that $H^s(M)$ is continuously embedded in $C^k(M)$ for all $s>k+n/2$. It follows that $\cap_{s\in\RR} H^s(M)=C^\infty(M)$.

 If $u\in \DD'(M)= \GoinX{\Omega_1}'$ is a distribution then, owing to compactness of $M$, there is a~partial differential operator $P$ so that $|u(\nu f)|\le \|Pf\|_\infty \le C_s \|Pf\|_{H^s(M)}$ for any $s>n/2$ and all $f\in C^\infty(M)$. Thus, with $t=s+\operatorname{ord} (P)$, $u(\nu f)= \big\langle \overline{f},w\big\rangle= \ip{T^t \overline{f}}{T^{-t}w}_{L^2(M)}$ for some $w\in H^{-t}(M)$, which is evidently compatible with the embedding of $H^{-t}(M)$ in $\DD'(M)$
 described earlier. In this sense,
 \begin{equation*}
 	\DD'(M) = \bigcup_{s\in\RR} H^{s}(M).
 \end{equation*}

 \subsubsection{General manifolds}

 For a general (not-necessarily compact) smooth manifold $M$ with at most finitely many components,
 we proceed as follows. If $K\subset M$ is compact and topologically regular,\footnote{Note that if $K$ is any compact subset, then the closure of the interior of $K$ is compact and topologically regular.} we may
 diffeomorphically identify $K$ with a compact subset $\hat{K}$ of a compact manifold $N$ (for example, by `doubling' a compact set that contains $K$ in its interior and has a smooth boundary~\cite{BaerWafo:2015}; for~$K$ contained in a coordinate chart one could take $N$ to be a torus). Then the Sobolev space~$H^s_K(M)$ is defined as the completion of
 $C^\infty_K(M)$ with respect to the pull back of (some choice of) the~$H^s(N)$ norm. The space $H^s_K(M)$ can be identified with a subspace of $\DD'_K(M)$, the distributions with support contained in $K$. As a topological vector space, it is independent of the choices used in its construction. However, by making such a choice one may
 endow $H^s_K(M)$ with a norm, denoted $\|\cdot\|_{H^s_K(M)}$, and indeed a compatible Hilbert space structure.
 Estimates of the form
 $\|f\|_{H^k_K(M)}\le C\|f\|_{K,k}$, $f\in C^\infty_K(M)$, and hence continuity of
 the embedding $C^k_K\to H^k_K$, carry over from the compact case for each fixed $k\in\NN_0$, as does compactness of the embedding $H^{s}_K(M)\to H^{t}_K(M)$ for $s>t$
 and the continuous embedding of $H^s_K(M)$ in $C^k_K(M)$ for $s>k+n/2$,
 where again $n$ is the maximum dimension of any component of $M$.
 Consequently one has $\cap_{s\in\RR} H^s_K(M)=C^\infty_K(M)$.

 Next, we define (abbreviating `compact and topologically regular' by c.t.r.)
 \begin{equation*}
 	H_0^s(M) = \bigcup_{\text{c.t.r. $K\subset M$}} H^s_K(M),
 \end{equation*}
 with the locally convex inductive limit topology.
 As $M$ admits countable compact exhaustions (and consequently, countable c.t.r.\ exhaustions), $H_0^s(M)$ may be realised as a countable strict inductive limit -- i.e., it is a $LF$ space. A fact to be used later is that, by~\cite[Proposition~14.6]{Treves:TVS} (see also Proposition~4 in~\cite{DieudonneSchwartz:1949}) a subset $B$ of $H_0^s(M)$ is bounded if and only if $B$ is a bounded subset of~$H^s_K(M)$ for some compact $K$. It is easily shown that
 \begin{equation*}
 	\EE'(M) = \bigcup_{s\in\RR} H_0^{s}(M).
 \end{equation*}

 Finally, the local Sobolev space $H^s_\loc(M)$ is defined as
 \begin{equation*}
 	H^s_\loc(M) = \big\{u\in\DD'(M)\colon \text{$\chi u\in H^s_0(M)$ for all $\chi\in C_0^\infty({M})$}\big\},
 \end{equation*}
 and is equipped with the Fr\'echet topology induced by the seminorms $\|\chi \cdot\|_{H^s_{K}(M)}$ as $K$ runs over compact topologically regular subsets of $M$ and $\chi$ runs over $C^\infty_K(M)$. The inclusion
 $H^s_0(M)\hookrightarrow H^s_\loc(M)$ is continuous for all $s$ and indeed we have
 \begin{equation*}
 	H^s_0(M) = H^s_\loc(M) \cap \EE'(M).
 \end{equation*}
 Thus, if $M$ is compact, $H^s_\loc(M)$ and $H^s_0(M)$ coincide as topological vector spaces. The construction above produces the same spaces as the chart-based approach taken in~\cite{Hormander3}.

 The Sobolev embedding theorems already mentioned entail the existence of continuous embeddings of $C^k(M)$ in $H^k_\loc(M)$ for all $k\in\NN_0$ and of $H^s_\loc(M)$ in $C^k(M)$ for all $s>k+n/2$, where~$n$ is the maximum dimension of any component of $M$. Consequently, $\cap_{s\in\RR} H^s_\loc(M)=C^\infty(M)$.

 \subsection{The topology of bounded convergence}\label{appx:top_bded_cvgnce}

 A general reference for the following is Tr\`eves~\cite[Chapter~32]{Treves:TVS}.

 If $E$ and $F$ are Hausdorff locally convex topological spaces then the topology of bounded convergence on $\LL(E,F)$, the space of all continuous linear maps $E\to F$, is defined by the neighbourhood base of zero, consisting of sets
 \begin{equation*}
 	\Uf(B;V) = \{T\in \LL(E,F)\colon T(B)\subset V\}
 \end{equation*}
 as $B$ runs over the bounded subsets of $E$ and $V$ runs over any neighbourhood base of zero in the topology of $F$. The notation $\LL_b(E,F)$ denotes $\LL(E,F)$ equipped with the topology of bounded convergence.
 Thus a net $T_\alpha$ converges to $0$ in $\LL_b(E,F)$
 if and only if, for every bounded $B\subset E$ and neighbourhood base set $V$,
 $T_\alpha$ eventually maps $B$ into $V$. This topology makes $\LL(E,F)$ Hausdorff and locally convex (inherited from $F$~\cite[p.~336]{Treves:TVS}). We record some basic facts.
 \begin{Lemma}\label{lem:Lb_basic}
 	Let $E$, $F$, $G$ be Hausdorff locally convex topological spaces. $(a)$ If $T\in \LL(E,F)$ and the net $S_\alpha\to S$ in $\LL_b(F,G)$ then $S_\alpha T\to ST$ in $\LL_b(E,G)$; $(b)$ if the net $T_\alpha\to T$ in $\LL_b(E,F)$ and $S\in \LL(F,G)$ then
 	$ST_\alpha\to ST$ in $\LL_b(E,G)$; $(c)$ if $T_\alpha\to T$ in $\LL_b(E,F)$ and $\Ran T_\alpha\subset \hat{F}$, where $\hat{F}$ is a topological subspace of $F$, then $T_\alpha\to T$ in $\LL_b(E,\hat{F})$; $(d)$ if $F$ is barrelled and $T_n\to T$ and $S_n\to S$ are convergent \emph{sequences} in $\LL_b(E,F)$ and $\LL_b(F,G)$, then $S_n T_n\to ST$ in~$\LL_b(E,G)$.
 \end{Lemma}
 Note that Hilbert, Banach, Fr\'echet and LF spaces are all barrelled~\cite[Chapter~33]{Treves:TVS}.
 \begin{proof}
 	$(a)$ It is enough to prove this in the case $S=0$; taking any bounded $B$ in $E$ with $0\in B$, and any $0$-neighbourhood $V$ in $G$, note that $T(B)$ is bounded in $F$ and deduce that $S_\alpha T(B)$ is eventually contained in $V$. Thus $S_\alpha T\to 0$ in $\LL_b(E,G)$. $(b)$ Without loss, suppose $T=0$ and with $B$ and $V$ as before, we note that $S^{-1}(V)$ is a $0$-neighbourhood in $F$ and again deduce that $ST_\alpha(B)$ is eventually contained in $V$, so $ST_\alpha\to 0$.
 	
 	$(c)$ Again, it is enough to treat $T=0$. Take any bounded $B\subset E$ and $0$-neighbourhood~$\hat{V}$ in~$\hat{F}$. Then $\hat{V}=\hat{F}\cap V$ for an $0$-neighbourhood~$V$ in $F$ because $\hat{F}$ carries the subspace topology. We know that, eventually, $T_\alpha(B)\subset V$, and as $\Ran T_\alpha\subset \hat{F}$, we have, eventually, that $T_\alpha(B)\subset V\cap\hat{F}= \hat{V}$. Thus $T_\alpha\to 0$ in $\LL_b(E,\hat{F})$.
 	
 	$(d)$ Suppose first that $T=0$ and let $0\in B$
 	be any bounded subset of $E$ and $V$ any $0$-neighbourhood in $G$. Then
 	the $S_n$, together with $S$, form a bounded subset in $\LL_b(F,G)$ (which, we recall, is Hausdorff). As $F$ is barrelled, they form an equicontinuous set of maps by the Banach--Steinhaus theorem~\cite[Theorem 33.1]{Treves:TVS}. Thus there exists
 	a $0$-neighbourhood $W$ in $F$ such that $S_n (W)\subset V$ for all $n$.
 	As $T_n(B)\subset W$ for all sufficiently large $n$, we have $S_nT_n(B)\subset W$ for such $n$. Thus $S_nT_n\to 0$ in this case. If $T\neq 0$, we now know
 	that $S_n(T_n-T)\to 0$, and as $S_n T\to ST$ we find $S_nT_n\to ST$ as required.
 \end{proof}

 Typical applications of Lemma~\ref{lem:Lb_basic}\,$(a)$,\,$(b)$ are to show, for instance, that a holomorphic function from a domain in $\CC$ to $\LL_b\big(H^s_0(M),H^s_K(M)\big)$ is also holomorphic as a function to $\LL_b\big(H^s_K\big)$ and
 $\LL_b(H^s_0(M))$ due to continuous embedding of $H_K^s(M)$ in $H^s_0(M)$. Part~$(d)$ has the following use.\looseness=-1
 \begin{Corollary}\label{cor:Leibniz}
 	Let $E$, $F$, $G$ be Hausdorff locally convex topological spaces, with $F$ being barrelled.
 	If $\lambda\mapsto T(\lambda)$ and $\lambda\mapsto S(\lambda)$ are functions from a domain in $\CC$ to~$\LL(E,F)$ and~$\LL(F,G)$ respectively, then $(a)$ if $T$ and $S$ are both continuous in the topology of bounded convergence $($at $\mu$$)$, so is $(ST)(\lambda)=S(\lambda)T(\lambda)$ $($at $\mu$$)$; $(b)$ if $T$ and $S$ are differentiable at $\mu$ in the topology of bounded convergence then so is $ST$, with a derivative given by the Leibniz rule
 	$(ST)'(\mu)=S'(\mu)T(\mu) + S(\mu)T'(\mu)$.
 \end{Corollary}
 \begin{proof}
 	As $\CC$ is first-countable, questions of continuity and differentiability may be reduced to sequential considerations, and the result follows using Lemma~\ref{lem:Lb_basic}\,$(d)$ and the standard proof of the Leibniz formula.
 \end{proof}

 It is also useful to have some sufficient conditions for convergence in $\LL_b(E,F)$ topology where~$E$ and $F$ are Fr\'echet or countable strict inductive limits thereof (LF spaces). If $F$ is Fr\'echet, then the bounded subsets are precisely the subsets $B$ so that $\sup_{f\in B} \rho_j(f)<\infty$ for every seminorm~$\rho_j$ defining the topology of $F$, while sets $V_{j,\epsilon}=\{f\in F\colon\rho_j(f)<\epsilon\}$, for arbitrary~$\epsilon>0$ and defining seminorm $\rho_j$, provide a basis of neighbourhoods of zero. On the other hand, if $E$ is an LF space with defining sequence $E_n$ of Fr\'echet spaces, then
 the bounded subsets~$B$ of~$E$ comprise precisely those subsets that are bounded subsets of some $E_n$~\cite[Proposition~14.6]{Treves:TVS}, while a neighbourhood base is provided by convex sets $V\subset E$ so that
 each~$V\cap E_n$ is an open neighbourhood of zero in $E_n$.
 We recall that continuous linear maps between topological vector spaces preserve boundedness~\cite[Proposition~14.2]{Treves:TVS}. The following results are used in the text.
 \begin{Lemma} \label{lem:LbFrechet}
 	Suppose $T_\alpha\in \LL_b(F)$ is a net of operators on a Fr\'echet space $F$ with defining seminorms $\rho_j$. If, for each $j$, there exists $k(j)$,
 	such that for all $\epsilon>0$, it is eventually true that $\rho_j(T_\alpha f)<\epsilon\rho_{k(j)}(f)$ for all $f\in F$, then $T_\alpha\to 0$ in $\LL_b(F)$.
 \end{Lemma}
 \begin{proof}
 	Given any bounded set $B$ and
 	neighbourhood $V_{j,\epsilon}$, we set $C=\sup_{f\in B}\rho_{k(j)}(f)$ and apply the given property to $\epsilon/C$ to find that, eventually,{\samepage
 	\begin{equation*}
 		\rho_j(T_\alpha f)<\epsilon C^{-1} \rho_{k(j)}(f) \le \epsilon, \qquad \forall f\in B.
 	\end{equation*}
 	We have shown that, eventually, $T_\alpha\in \Uf(B;V_{j,\epsilon})$; as $j$ and $\epsilon$ were arbitrary, $T_\alpha\to 0$ in $\LL_b(F)$.}
 \end{proof}

 \begin{Corollary} \label{cor:HstoCinfty}
 	Consider a net $T_\alpha\in \LL\big(C^\infty_K\big)$
 	such that each $T_\alpha$ extends to an operator in~$\LL(H^s_K)$ $($which we also denote $T_\alpha$$)$ for all $s\ge s_*$. Suppose also that $T_\alpha\to 0$ in every $\LL(H^s_K)$, $s\ge s_*$. Then $T_\alpha\to 0$ in $\LL_b\big(C^\infty_K\big)$.
 \end{Corollary}
 \begin{proof}
Any $C^\infty_K$ seminorm is bounded by an $H^s_K$ norm and vice versa, so, for any
 	$C^\infty_K$-seminorm $\|\cdot\|_{K,j}$ we may estimate
 	\begin{equation*}
 		\|T_\alpha f\|_{K,j} \le c \| T_\alpha f\|_{H^{s(j)}_K} \le c\|T_\alpha\|_{\LL(H^{s(j)}_K)} \|f\|_{H^{s(j)}_K}
 		\le c' \|T_\alpha\|_{\LL(H^{s(j)}_K)} \|f\|_{K,k(s(j))}
 	\end{equation*}
 	for suitable choices of $s(j)$ and $k(s(j))$, uniformly in $f\in C^\infty_K$ and $\alpha$.
 	As $T_\alpha\to 0$ in $\LL\big(H^{s(j)}_K\big)$, for any $\epsilon>0$ it is eventually true that
 	$\|T_\alpha f\|_{K,j} \le \epsilon \|f\|_{K,k(s(j))}$ for all $f\in C^\infty_K$. Hence $T_\alpha\to 0$ in $\LL_b\big(C^\infty_K\big)$ by Lemma~\ref{lem:LbFrechet}.
 \end{proof}

 \begin{Lemma}\label{lem:LbLF}
 	Suppose $T_\alpha$ is a net of operators on the LF space $E=\bigcup_{n\in\NN} E_n$. If for some fixed $n$, one has $\Ran(T_\alpha)\subset E_n$ and $T_\alpha|_{E_m}\to 0$ in $\LL_b(E_m,E_n)$ for all $m$, then $T_\alpha\to 0$ in both $\LL_b(E,E_n)$ and $\LL_b(E)$.
 \end{Lemma}
 \begin{proof}
 	Consider any neighbourhood $V$ of zero in the standard neighbourhood base of $E_n$ and any bounded set $B\subset E$, necessarily obeying $B\subset E_m$ for some $m$. Then $T_\alpha$ eventually maps $B$ into $V$; as $B$ and $V$ were arbitrary, $T_\alpha\to 0$ in $\LL_b(E,E_n)$. Post-composing with the continuous embedding of $E_n$ in $E$, $T_\alpha\to 0$ in $\LL_b(E)$ as well.
 \end{proof}

\subsection*{Acknowledgements}
 	It is a pleasure to thank Rainer Verch for useful discussions at various stages of this work and Maximilian Ruep for a careful reading and comments on a draft of the manuscript. I~particularly thank Christian B\"ar for asking the question that prompted Theorem~\ref{thm:Fredholm} and also Lashi Bandara and other participants of the conference `Global analysis on manifolds' held in B\"ar's honour (Freiburg, September 2022) for useful remarks and conversations. I also thank the referees for valuable suggestions.

\pdfbookmark[1]{References}{ref}
\LastPageEnding

\end{document}